\documentclass[runningheads]{llncs}
\usepackage{geometry}
\usepackage{graphicx}
\usepackage[numbers]{natbib}

\usepackage[english]{babel}
\usepackage{physics, amsmath, amssymb, mathtools, comment, subcaption}
\usepackage{algorithm}
\usepackage{algorithmicx}
\usepackage{algpseudocode}
\usepackage{adjustbox}
\makeatletter
\newcommand{\@chapapp}{\relax}%
\makeatother

\usepackage[title,toc,titletoc,header]{appendix}
\usepackage{hyperref}

\title{A Novel Skip Orthogonal List for Dynamic Optimal Transport Problem}
\author{
    Xiaoyang Xu,
    Hu Ding\thanks{Corresponding Author}
}
\authorrunning{X. Xu, H. Ding}
\institute{University of Science and Technology of China\\
\email{xiaoyangxu@mail.ustc.edu.cn,
huding@ustc.edu.cn}}

\begin{document}

\maketitle
\begin{abstract}
  Optimal transport is a fundamental topic that has attracted a great amount of attention from the optimization community in the past decades. In this paper, we consider an interesting  discrete dynamic optimal transport problem:
   can we efficiently update the optimal transport plan when the weights  or the locations of the data points  change? 
  This problem is naturally motivated by several applications in  machine learning. For example, we often need to compute the optimal transport cost between two different data sets; if some changes happen to a few data points, should we re-compute the high complexity cost function or update the cost by some efficient dynamic data structure? We are aware that several dynamic maximum flow algorithms have been proposed before, however, the research on dynamic minimum cost flow problem is still quite limited, to the best of our knowledge.
  We propose a novel 2D Skip Orthogonal List together with some dynamic tree techniques. 
  Although our algorithm is based on the conventional simplex method,
  it can efficiently find the variable to pivot within expected $O(1)$ time,
  and complete each pivoting operation within expected $O(\abs{V})$ time where $V$ is the set of all supply and demand nodes.
  Since dynamic modifications typically do not introduce significant changes,
  our algorithm requires only a few simplex iterations in practice.
  So our algorithm is 
  more efficient than re-computing the optimal transport cost that needs at least one traversal over all $\abs{E} = O(\abs{V}^2)$ variables,
  where $\abs{E}$ denotes the number of edges in the network.
  Our experiments demonstrate that our algorithm significantly outperforms existing algorithms in the dynamic scenarios.
\end{abstract}

\section{Introduction}

The discrete {\em optimal transport (OT)} problem involves finding the optimal transport plan  ``$X$'' that minimizes the total cost of transporting one weighted dataset $A$  to another $B$, given a cost matrix ``$C$''~\cite{peyre2019computational}. The datasets $A$ and $B$ respectively represent the supply and demand node sets, and the problem can be represented as a minimum cost flow problem by adding the edges between $A$ and $B$ to create a complete bipartite graph. The discrete optimal transport problem finds numerous applications in the areas such as image registration~\cite{haker2004optimal}, seismic tomography~\cite{metivier2016measuring}, and machine learning~\cite{torres2021survey}. However, most of these applications only consider static scenario where the weights of the datasets and the cost matrix remain constant. 
Yet, many real-world applications need to consider the dynamic scenarios:

\begin{itemize}
    \item \textbf{Dataset Similarity}. In data analysis, measuring the similarity between datasets is a crucial task, and optimal transport has emerged as a powerful tool for this purpose~\cite{alvarez-melis2020geometric}. Real-world datasets are often dynamic, with data points being replaced, weights adjusted, or new data points added over time. Therefore, it is necessary to take these dynamically changes  into account.
    \item \textbf{Time Series Analysis}. Optimal transport can serve as a metric in time series analysis~\cite{cheng2021dynamical}. The main intuition lies in the smooth transition of states between time points in a time series. The smoothness implies the potential to iteratively refine a new solution based on the previous one, circumventing the need for a complete recomputation.

    \item \textbf{Neuroimage analysis}~\cite{gramfort2015fast, janati2019group}. In the medical imaging applications, we may want to compute the change trend of a patient's organ (e.g., the MRI images of human brain over several months), and the differences are measured by the optimal transport cost. Since the changes are often local and small, we may hope to apply some efficient method to quickly update the cost over the period. 

\end{itemize}

Denote by $V$ and $E$ the sets of vertices and edges in the bipartite network, respectively. 
Existing methods, such as the Sinkhorn algorithm~\cite{cuturi2013sinkhorn} and the Network Simplex algorithm~\cite{orlin1997polynomial}, are not adequate 
 to handle the dynamic scenarios. Upon any modification to the cost matrix, the Sinkhorn algorithm requires at least one Sinkhorn-Knopp iteration to regularize the entire the solution matrix, while the Network Simplex algorithm needs to traverse all edges at least once. Consequently, the time complexities of these algorithms for the dynamic model are $\Omega(\abs{V}^2)$ in general cases.

Our algorithm takes a novel data structure that yields an \( O(s|V|) \) time solution for handling evolving datasets, where \( s \) is determined by the magnitude of the modification. In practice, \( s \) is usually much less than the data size $|V|$, and therefore our algorithm can save a large amount of runtime for the dynamic scenarios.
\footnote{Demo library is available at \url{https://github.com/xyxu2033/DynamicOptimalTransport}}

\subsection{Related Works}
\textbf{Exact Solutions.} 
In the past years, several linear programming based minimum cost flow algorithms have been proposed to address discrete optimal transport problems. The simplex method by \citet{dantzig1955generalized} can efficiently solve  general linear programs. Despite its worst-case exponential time complexity, \citet{spielman2004smoothed} showed that its smoothed time complexity is polynomial. \citet{cunningham1976network} adapted the simplex method for minimum cost flow problems. Further, \citet{orlin1997polynomial} enhanced the network simplex algorithm with cost scaling techniques and \citet{tarjan1997dynamic} improved its complexity to be $\tilde O(\abs{V}\abs{E})$. Recently, \citet{van2021minimum} presented an algorithm based on the interior point method with a time complexity  $O\qty(\abs{E} + \abs{V}^{1.5})$, and \citet{chen2022maximum} proposed a $\abs{E}^{1 + o(1)}$ time algorithm through a specially designed data structure on the interior point method.

\textbf{Approximate Algorithms.} For approximate optimal transport, \citet{sherman2017generalized} proposed a $(1 + \epsilon)$ approximation algorithm in $\epsilon^{-2}|E|^{1 + o(1)}$ time.  \citet{pele2009fast} introduced the FastEMD algorithm that applies classic algorithms on a heuristic sketch of the input graph. Later, \citet{cuturi2013sinkhorn} used Sinkhorn-Knopp iterations to approximate the optimal transport problem by adding the smoothed entropic entry as the regularization term. Recently several optimizations on the Sinkhorn algorithm have been proposed, such as the Overrelaxation of  Sinkhorn~~\cite{thibault2021overrelaxed} and the Screening Sinkhorn algorithm~~\cite{alaya2019screening}.


\textbf{Search Trees and Skip Lists.} Our data structure also utilizes high-dimensional extensions of skip lists to maintain a 2-dimensional Euler Tour sequence. Existing high-dimensional data structures based on self-balanced binary search trees, such as $k$-d tree~\cite{bentley1975multidimensional}, are not suitable as they do not support cyclic ordered set maintenance. Skip lists~\cite{pugh1990skip} as depicted in Figure \ref{fig:skiplistoverview}, which are linked lists with additional layers of pointers for element skipping, is adapted in our context to form skip orthogonal lists.
This skipping technique is later generalized for sparse data in higher dimension~\cite{nickerson1998skip,eppstein2005skip}, but range querying generally requires $O(n^2)$ time where $n^2$ is the number of points in the high dimensional space. On the other hand, our data structure requires expected $O(n)$ time when applied to simplex iterations.

\subsection{Overview of Our Algorithm}

Our algorithm for the dynamic optimal transport problem employs two key strategies:

\textbf{First, the dynamic optimal transport operations are reduced to simplex iterations.}
Our technique, grounded on the Simplex method, operates by eliminating the smallest cycle in the graph. We assume that the modifications influence only a small portion of the result, requiring only a few simplex iterations. It is worth noting that existing algorithms like the Network Simplex Algorithm perform poorly under dynamic modifications as they require scanning all the edges at least once to ensure the correctness of the solution.

\textbf{Second, an efficient data structure is proposed for performing each simplex iteration within expected linear time complexity.}
Our data structure, as shown in Figure \ref{fig:2dskiplistoverview}, employs the Euler Tour Technique. We adapt skip lists to maintain the cyclic ordered set produced by the Euler Tour Technique and introduce an additional dimension to create a Skip Orthogonal List. This structure aids in maintaining the information about 
the \textit{adjusted cost matrix}, which is a matrix that requires specific range modifications and queries.

\begin{figure}[htbp]
\centering
\begin{subfigure}[b]{0.48\linewidth}
    \includegraphics[width = \linewidth]{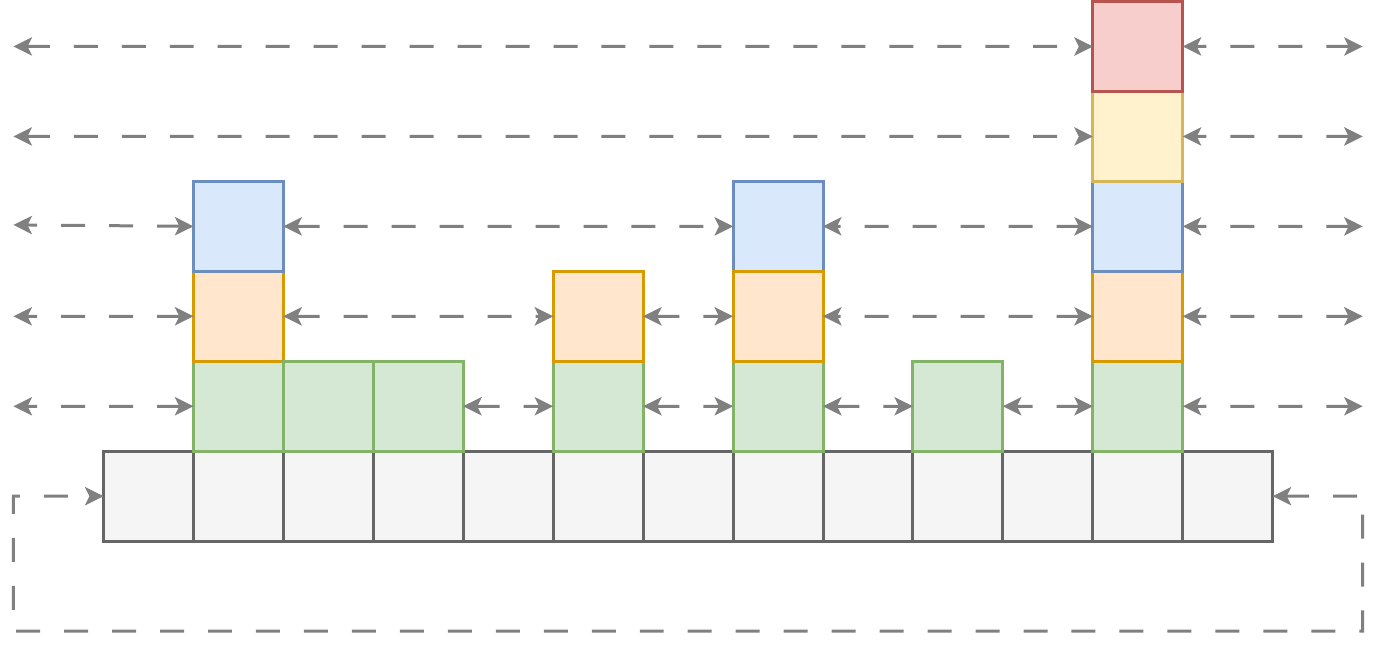}
    \caption{An example for 1D Euler Tour Tree}
    \label{fig:skiplistoverview}
  \end{subfigure}\hfill
  \begin{subfigure}[b]{0.5\linewidth}
      \includegraphics[width = \linewidth]{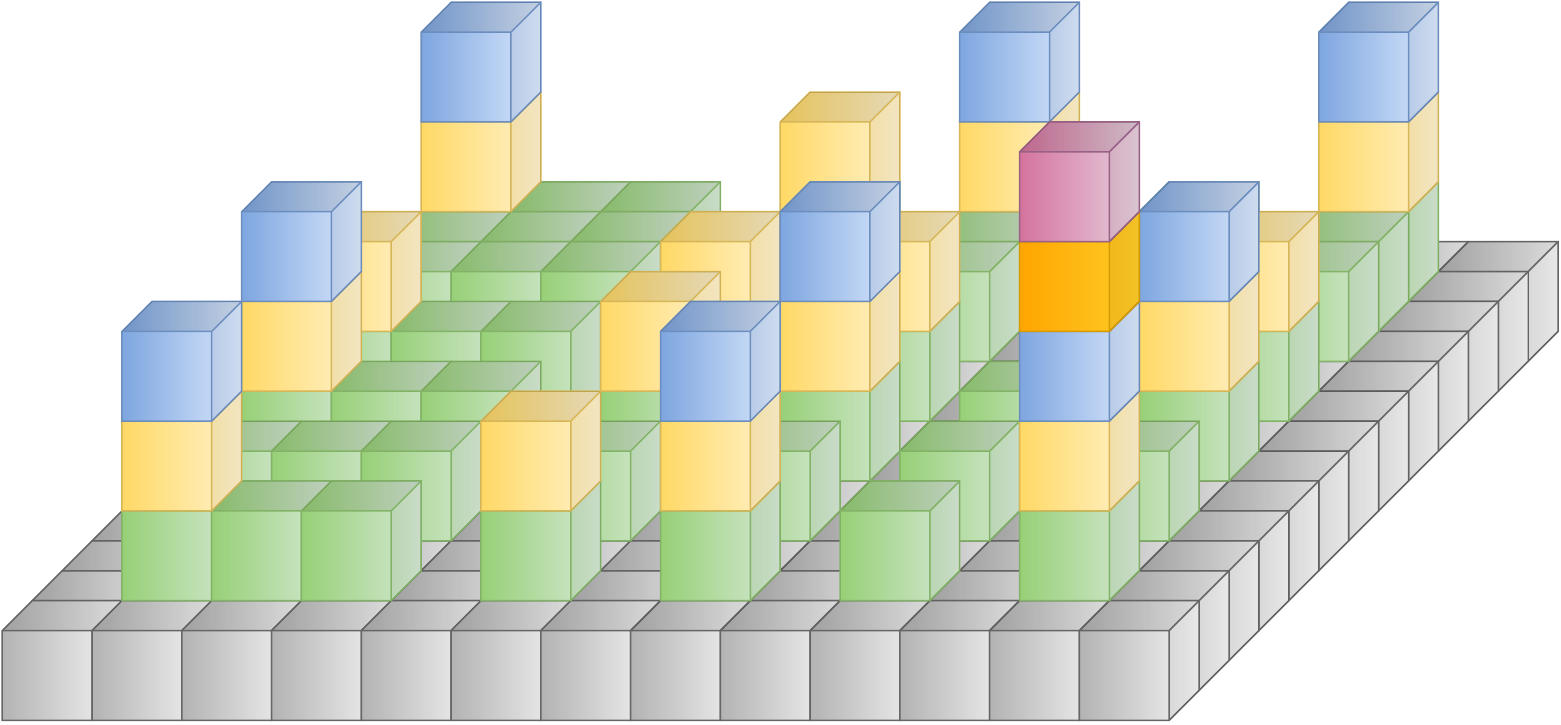}
    \caption{Our 2D Euler Tour Tree}
    \label{fig:2dskiplistoverview}
\end{subfigure}
\caption{Overview of Euler Tour Tree with Skip Lists}
\end{figure}

The rest of the paper is organized as follows. In Section~\ref{sec:pre}, we introduce several important definitions and notations that are used throughout this paper. 
 In Section~\ref{sec:datastructure} we present the data structure \textit{Skip Orthogonal List},
where subsection~\ref{sec:overallstructure} explains how the data structure is organized 
and subsection~\ref{sec:cut} uses \textit{cut} operation as an example to demonstrate  the updates on this data structure.
In Section~\ref{sec:simplex} we elaborate on how to use our data structure to solve the dynamic optimal transport problem. 
Subsection~\ref{sec:dynamicmodel} shows that the dynamic optimal transport model can be reduced to simplex iterations,
and Subsection~\ref{sec:algorithm} shows how our data structure could be used to improve the performance of each simplex iteration.

\section{Preliminaries}
\label{sec:pre}
\subsection{Optimal Transport}
\label{sec:ot}

Let $A$ and $B$ represent the source and target point sets, respectively;  the corresponding discrete probability distributions are $\alpha \in \mathbb R_{\geq 0}^A$ and $\beta\in \mathbb R_{\geq 0}^B$, such that $\sum_{a\in A}\alpha_a = \sum_{b\in B}\beta_b = 1$. The cost matrix is $C \in \mathbb R^{A\times B}$ with each entry $c_{ab}$ denoting the cost of transporting a unit of mass from the point $a\in A$ to the point $b\in B$. The discrete optimal transport problem can be formulated as \eqref{eqn:originalproblem}.

\begin{equation}
  \label{eqn:originalproblem}
  \begin{aligned}
\mathcal W
(\alpha, \beta, C) \triangleq &\min_{X \in \mathbb R_{\geq 0}^{A\times B}} \sum_{a \in A}\sum_{b \in B} c_{ab}x_{ab}\\ \text{subject to} \quad &\begin{cases}\sum_{b\in B}x_{ab} = \alpha_a & \forall a\in A\\\sum_{a\in A}x_{ab} = \beta_b&\forall b\in B\end{cases}
\end{aligned}
\end{equation}

Since Problem \eqref{eqn:originalproblem} is a standard network flow problem, it can be transformed to the following Problem \eqref{eqn:fullproblem} by adding infinity-cost edges~\cite{peyre2019computational}:
\begin{equation}
  \label{eqn:fullproblem}
  \begin{aligned}
  \mathcal W(w, C) \triangleq &\min_{X \in \mathbb R_{\geq 0}^{V\times V}} \sum_{u \in V}\sum_{v \in V} c_{uv}x_{uv} \\ \text{subject to} \quad &\sum_{v\in V}x_{uv} - \sum_{v\in V}x_{vu} = w_u \quad \forall u\in V
  \end{aligned}
\end{equation}
\begin{equation}
\label{eqn:infedge}
  c_{a_1a_2} = c_{b_1a_1}\!\! = c_{b_1b_2} = +\infty \quad \!\!\!\forall a_1, a_2\in A, b_1, b_2\in B
  \end{equation}
  \begin{equation}
  \label{eqn:redefpointweight}
  w_{u} = \begin{cases}\alpha_a & \text{if } u = a\in A\\-\beta_b & \text{if } u = b\in B\end{cases}
\end{equation}

We add the constraint \eqref{eqn:infedge}, and also redefine the point weights as \eqref{eqn:redefpointweight}. Note that in \eqref{eqn:fullproblem}, the input weight $w$ must always satisfy $\sum_{v\in V}w_v = 0$; otherwise, the constraints cannot be satisfied.

We use  $X$ to denote a given basic solution in the context of using the simplex method to solve the Optimal Transport problem.
We notice that the basic variables always form a spanning tree of the
complete directed graph with self loops $G(V, V\times V)$
\cite{cunningham1976network}.
Let the dual variables be $\pi\in \mathbb R^V$, satisfying the following constraint:
\begin{equation}
x_{uv} \text{ is a basic variable } \implies \pi_u - \pi_v = c_{uv}.
\label{eqn:complementaryslackness}
\end{equation}

We then define the adjusted cost matrix $C^\pi \in \mathbb R^{V\times V}$ as $c^\pi_{uv} \triangleq c_{uv} + \pi_v - \pi_u$,
where $c^\pi_{uv}$ represents the simplex multipliers for the linear program~\cite{orlin1997polynomial}.

\subsection{Euler Tour Technique}
\label{sec:ett}
The {\em Euler Tour Technique} is a method for representing a tree $T$ as a cyclic ordered set of linear length~\cite{tarjan1997dynamic}.
Specifically, given a tree $T(V, E_T)$, we construct a directed graph $D_T(V, E_D)$ with the same vertex set as follows:
\begin{itemize}
\item For each vertex $v \in V_T$, add the self-loop $(v, v)$ to $E_D$;
\item For each undirected edge $\qty(u, v) \in E$, add two directed edges $(u, v)$ and $(v, u)$ to $E_D$.
\end{itemize}

Following this definition, $\abs{E_D} = 2\abs{E_T} + \abs{V_T}$. Since $T$ is a tree, $\abs{E_T} = \abs{V} - 1$, and therefore $\abs{E_D} = 3\abs{V} - 2 = O\qty(\abs{V})$. 
Since the difference of In-Degree and Out-Degree of each vertex in $D_T$ is 0, $D_T$ should always contain an Euler Tour. 

\begin{definition}[Euler Tour Representation]
  Given a tree $T$, the Euler Tour representation is an arbitrary sequence of Euler Tour of $D_T$ represented by edges. That is, $E_D$ with circular order induced by the Euler Tour is an Euler Tour representation.
\label{dfn:eulertourrepresent}
\end{definition}

For the rest of the paper, $E_D$ denotes the Euler Tour representation of the tree in the context.

Through Definition \ref{dfn:eulertourrepresent},
we can reduce \textit{edge linking}, \textit{edge cutting}, \textit{sub-tree weight updating} and \textit{sub-tree weight querying} to constant number of \textit{element insertion}, \textit{element deletion}, \textit{range weight modification} and \textit{range weight querying} on a circular ordered set~\cite{tarjan1997dynamic}.
We show in Section~\ref{sec:dynamicmodel} that the dynamic optimal transport can be reduced to the 2D version of 
these four operations. 

\subsection{Orthogonal Lists and Skip Lists}

{\em Skip Lists} are the probabilistic data structures that extend a singly linked list with forward links at various levels, for improving the search, insertion, and deletion operations. Figure~\ref{fig:skiplistoverview} illustrates an example for skip list. Each level contains a circular linked list, where the list at a higher level is a subset of the list at a lower level and the bottom level contains all the elements. The nodes at the same level have the same color and are linked horizontally. The corresponding elements in adjacent lists are connected by vertical pointers. We apply this skipping technique to circular singly linked lists in our work.
Just as most self-balanced binary search trees, Skip Lists support ``lazy propagation'' techniques, allowing range modifications within $O(\log n)$ time, where $n$ is the sequence length maintained by the tree~\cite{sleator1985self}. This technique is commonly used in dynamic trees for network problems~\cite{tarjan1997dynamic}.

A $k$-dimensional {\em Orthogonal List} has $k$ orthogonal forward links (it is a standard linked list when $k = 1$). Orthogonal lists, which can be singly, doubly, or circularly linked, can maintain the information mapped from the Cartesian product of ordered sets, such as sparse tensors~\cite{butterfield2016dictionary}. 
Figure~\ref{fig:orthogonallist}
demonstrates an orthogonal list that maintains a $3\times 4$ matrix. Each node has 2 forward links,
denoted by \textit{row} links (red) and \textit{column} links (blue).
The \textit{row} links connect the elements in each row into a circular linked list horizontally and the \textit{column} links connect the elements in each column into a circular linked list vertically.

\noindent
\begin{figure}[htbp]
  \centering
  \begin{minipage}{0.45\linewidth}
  \includegraphics[width = \linewidth]{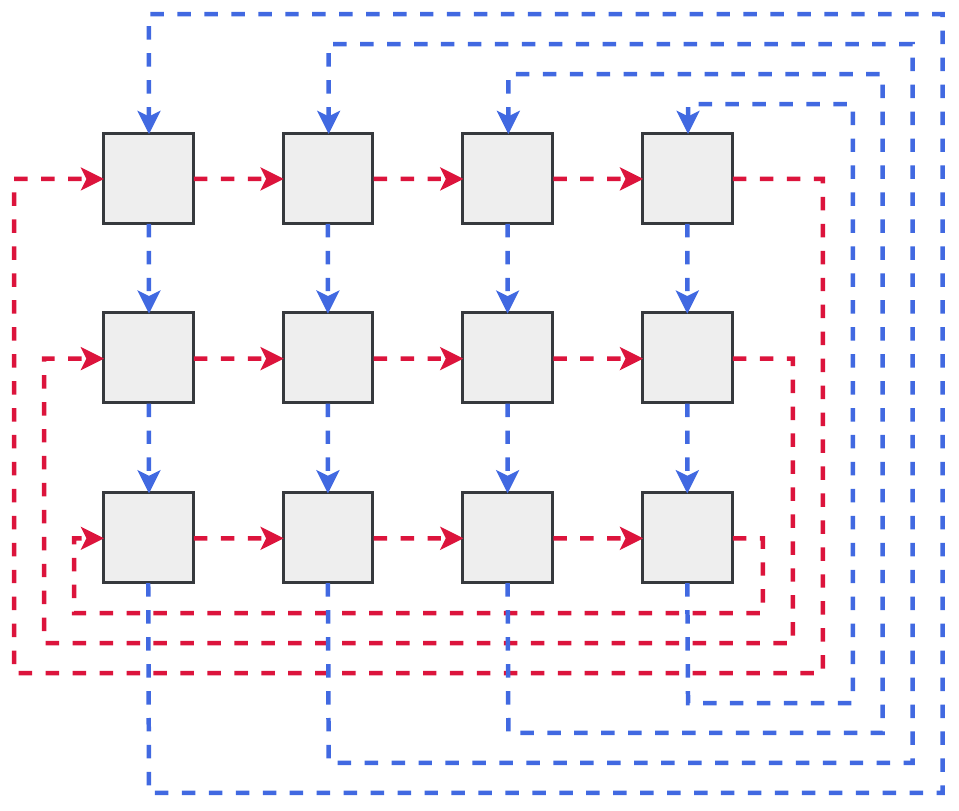}
    \caption{An Orthogonal List Example.}
    \label{fig:orthogonallist}
  \end{minipage}\hfill
  \begin{minipage}{0.55\linewidth}
      \includegraphics[width = \linewidth]{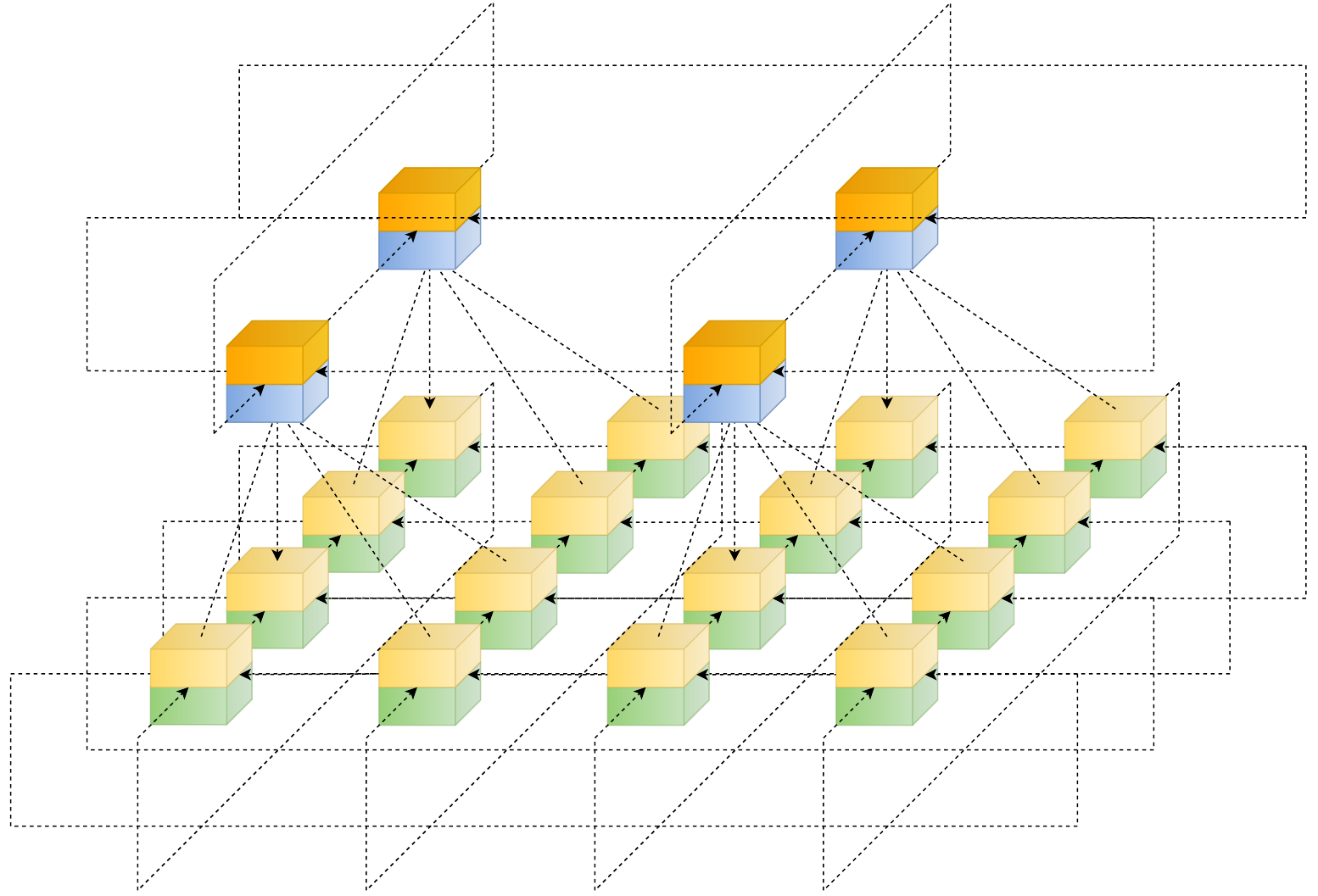}
    \caption{An illustration for Skip Orthogonal List.}
    \label{fig:2dskiplistwithlink}
    \end{minipage}
  \end{figure}

\section{Skip Orthogonal List}
\label{sec:datastructure}
In this section~we introduce our novel data structure \textbf{Skip Orthogonal List}.
In Section~\ref{sec:simplex}, this data structure is used for dynamically updating optimal transport. 
Formally, with the help of a Skip Orthogonal List,
we can maintain a forest
$F(V_F, E_F)$
with at most two trees,
and a matrix $C^\pi \in \mathbb R^{V_F\times V_F}$ that supports the following operations. The first two and the last operations are for the case that $F$ contains only one tree; the other two operations are for the case that $F$ has two trees. 
\begin{itemize}
\item \textbf{Cut}. 
Given an undirected edge $\qty(u, v)$, remove edge $(u, v)$ from the tree, and split it into two disjoint trees. Let the connected component containing $u$ form the vertex set $V_1$, and the connected component containing $v$ form the vertex set $V_2$.
\item \textbf{Insert}. 
Add a new node to $F$ that does not connects with any other node. Let the original nodes form the vertex set $V_1$, and the new node itself form the vertex set $V_2$.
\item \textbf{Range Update}. 
Given $x$, for each $(u, v)\in V\times V$,
update $c^\pi_{uv}$ as equation \eqref{eqn:rangeadd}.
\begin{align}
\label{eqn:rangeadd}
c^\pi_{uv} \gets c^\pi_{uv} + 
\begin{cases}
0 & (u, v)\in V_1\times V_1 \\
-x & (u, v)\in V_1\times V_2\\
x & (u, v)\in V_2\times V_1\\
0 & (u, v)\in V_2\times V_2
\end{cases}
\end{align}
\item \textbf{Link}. 
Given a pair $\{u, v\}$, add the edge $(u, v)$ to the forest; connect two disjoint trees into a single tree, if  $u$ and $v$ are disconnected.
\item \textbf{Global Minimum Query}. 
Return the minimum value of $C^\pi$ on the tree.
\end{itemize}

For the remaining of the section,
we construct a data structure with the expected $O(\abs{V}^2)$  space complexity, where each operation can be done with the expected $O(\abs{V})$ time.
Section~\ref{sec:overallstructure} shows the overall structure of the data structure and how to query in this data structure. Section~\ref{sec:cut} illustrates the cut operation as an example based on this structure. For other operations (linking, insertion, and range updating),
we leave them to appendix \ref{sec:timeproof}.

\subsection{The Overall Structure}
\label{sec:overallstructure}
As shown in Figure~\ref{fig:2dskiplistwithlink}, a \textbf{Skip Orthogonal List} is a hierarchical 
collection of Orthogonal lists, 
where each layer has fewer elements than
the one below it, and the elements are evenly
spaced out.
The bottom layer has all the elements
while the top layer has the least. Formally, it can be defined as Definition \ref{dfn:2dskiplist}.

\begin{definition}
\label{dfn:2dskiplist}
Given a parameter $p$ and a cyclic ordered set $S$,
a \textbf{2D Skip Orthogonal List} $\mathcal L$ over the set $S$ is an infinite collection of 2 Dimensional Circular Orthogonal Lists $\mathcal L = \{L_0, L_1, \cdots\}$, where
\begin{itemize}
    \item $\qty{h(s)}_{s\in S}$ is a set of $\abs{S}$ independent random variables.
    The distribution is a geometric distribution with parameter $p$
    \item For each $\ell\in \mathbb N$,
    let $L_\ell$ be an Orthogonal List whose key contains all the elements in 
    $S_\ell\times S_\ell$,
    where $S_\ell$ is the cyclic sequence formed by $S_\ell\triangleq \qty{s\in S\mid h_s \geq \ell}$
\end{itemize}
\end{definition}

Note that for any pair $u, v\in S$, we use $(u, v)$ to denote the corresponding element in $S\times S$; with a slight abuse of notations, we also use  ``$(u, v)$ at level $\ell$'' to denote the corresponding node  at the $\ell$-th level in the Skip Orthogonal List. If level $\ell$ is not specified in the context,  $(u, v)$ refers to the node at the bottom level.

We use this data structure to maintain several key information of $E_D\times E_D$.
Since $\abs{E_D} = O(\abs{V})$ as discussed in Section~\ref{sec:ett},
similar to conventional 1D Skip Lists,
we know that the space complexity is $O(\abs{V}^2)$ with high probability
in appendix \ref{sec:spaceproof}.

Now we augment this data structure to store some additional information for range updating and global minimum query.
Before that, the concept ``dominate" needs to be adapted to 2D case defined as Definition \ref{dfn:2ddominate}.

\begin{definition}
\label{dfn:2ddominate}
For any positive integer $n$, in a Skip Orthogonal List $\mathcal L$ over the cyclic ordered set $S$,
suppose $(u_1, v_1)$ and $(u_2, v_2)$ are 2 elements in $S\times S$.
We say the node
$(u_1, v_1)$
\textbf{dominates} $(u_2, v_2)$ at level $\ell$ if and only if the following three conditions are all satisfied:
\begin{itemize}
    \item $h(u_1)\geq l$ and $h(v_1)\geq \ell$;
    \item $u_1 = u_2$ or
    $\max_{u = u_1 + 1}^{u_2}h(u) < \ell$;
    \item $v_1 = v_2$ or
    $\max_{v = v_1 + 1}^{v_2}h(v) < \ell$.
\end{itemize}
Here, for any element $s$ in the cyclic ordered set $S$, we use ``$s + 1$'' to denote the successor of $s$ induced by the cyclic order.
\end{definition}

To better understand Definition~\ref{dfn:2ddominate}, we illustrate the examples  in Figure~\ref{fig:dominate} and Figure~\ref{fig:2ddominate}. In each figure,
each blue node dominates itself and all the yellow nodes,
while the red nodes dominate every node in the orthogonal list.

\begin{figure}[htbp]
  \centering
\begin{subfigure}[b]{0.45\linewidth}
\centering
  \includegraphics[width = \linewidth]{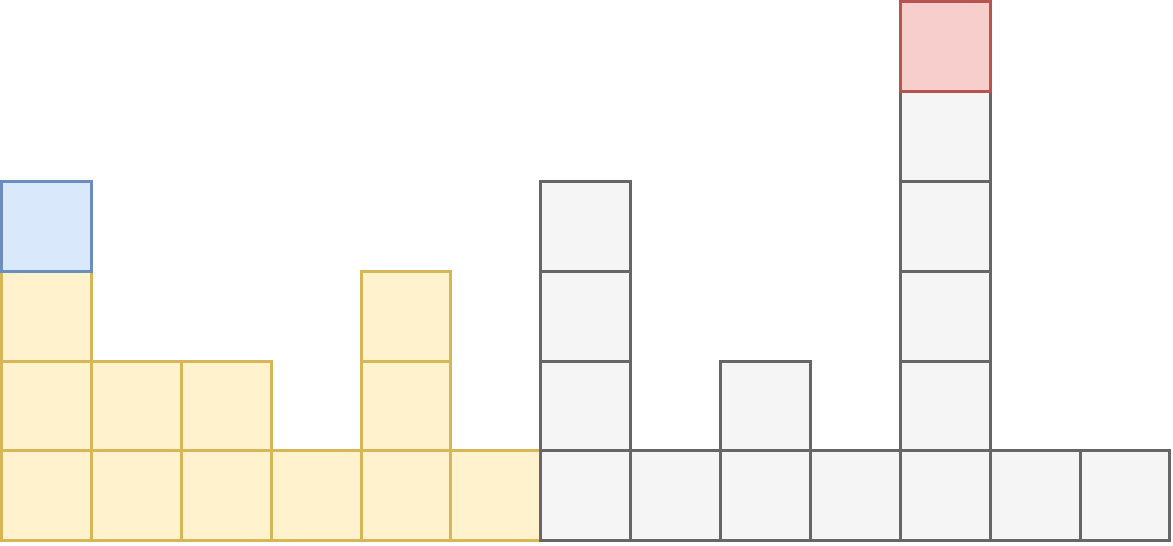}
    \caption{Domination in conventional 1D Skip List}
    \label{fig:dominate}
  \end{subfigure}\hfill
  \begin{subfigure}[b]{0.5\linewidth}
    \centering
    \includegraphics[width = \linewidth]{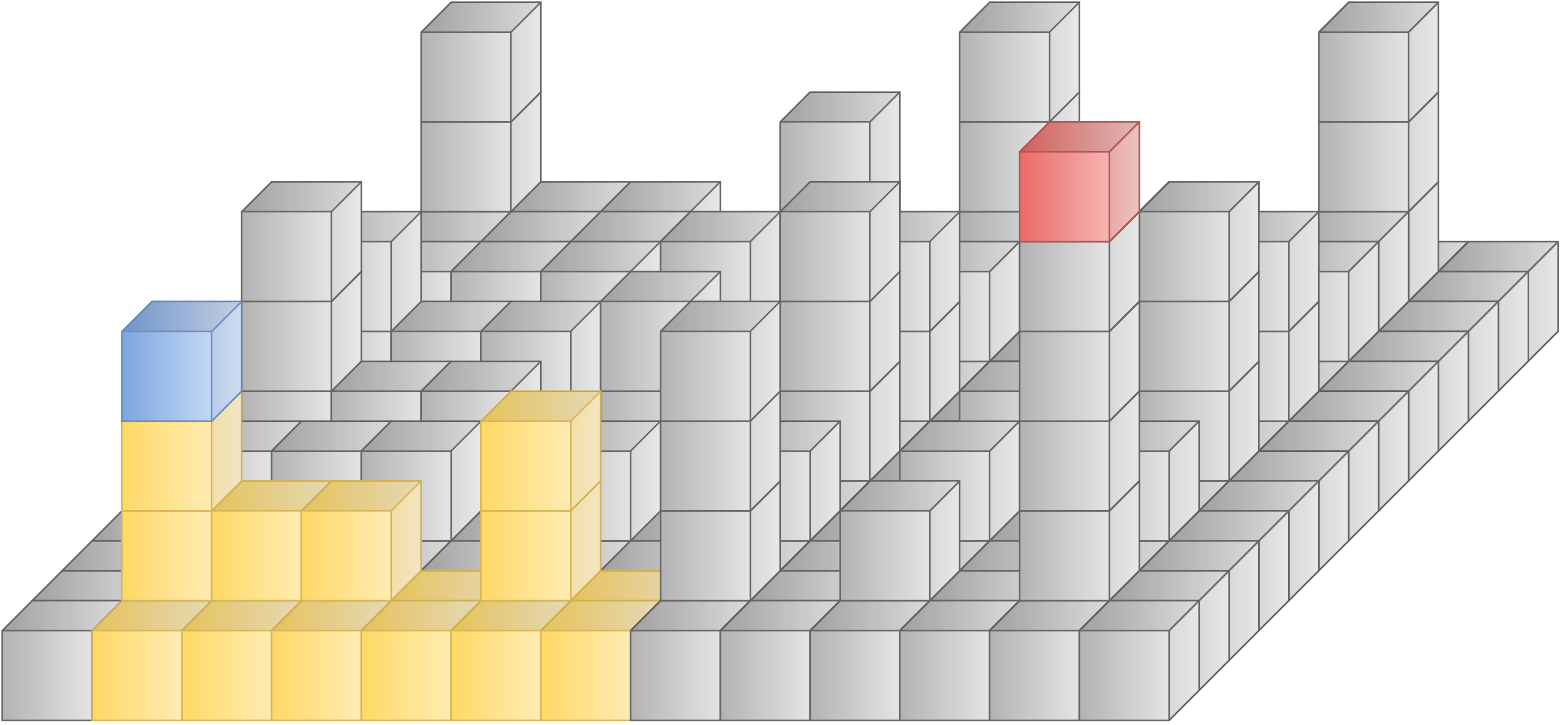}
    \caption{Domination in 2D Skip List}
    \label{fig:2ddominate}
\end{subfigure}
\caption{Illustrations of \textit{Domination} in Skip Lists}
\end{figure}

We now augment the Skip Orthogonal List of
Definition~\ref{dfn:2dskiplist}.
For each node $(u, v)$ at orthogonal list $L_\ell$,
beside the two forward links and two backward links,
we add the following attributes:

\begin{itemize}
    \item \textbf{tag}:  maintains 
    the tag for lazy propagation
    for all the nodes dominated by it;
    \item \textbf{min\_value}:  maintains
    the minimum value
    among all the nodes dominated by it. Note that when $\ell = 0$, it stores the original value of $c^\pi_{uv}$ following lazy propagation technique.
    That is, for any node $x$ in the data structure, after each modification and query, the data structure needs to assure
\begin{align*}
&x.\textit{min\_value} + \sum_{y \textit{ dominates } x}y.\textit{tag}\\
=\ &\min_{(u, v)\in S\times S, x \textit{ dominates } (u, v)} c^\pi_{uv}.
\end{align*}
    \item \textbf{min\_index}: maintains the index corresponding to \textit{min\_value} attribute.
    \item \textbf{child}:  points to $(u, v)$ at the orthogonal list $L_{\ell - 1}$ if $\ell > 0$, and it is invalid if $\ell = 0$;
    \item \textbf{parent}:  points to the node that dominates it if $L_{\ell + 1}$ is not empty.
\end{itemize}

\subsection{The Update Operation: Cut}
\label{sec:cut}
In this subsection, we focus on the update operation ``Cut'' for a 2D Skip Orthogonal List as an example.
Figure~\ref{fig:2dcutflatten} and Figure~\ref{fig:3dcut} illustrate the basic idea of the cutting process.

\begin{figure}[htbp]
  \centering
  \begin{subfigure}[b]{0.35\linewidth}
    \centering
    \includegraphics[width=\linewidth]{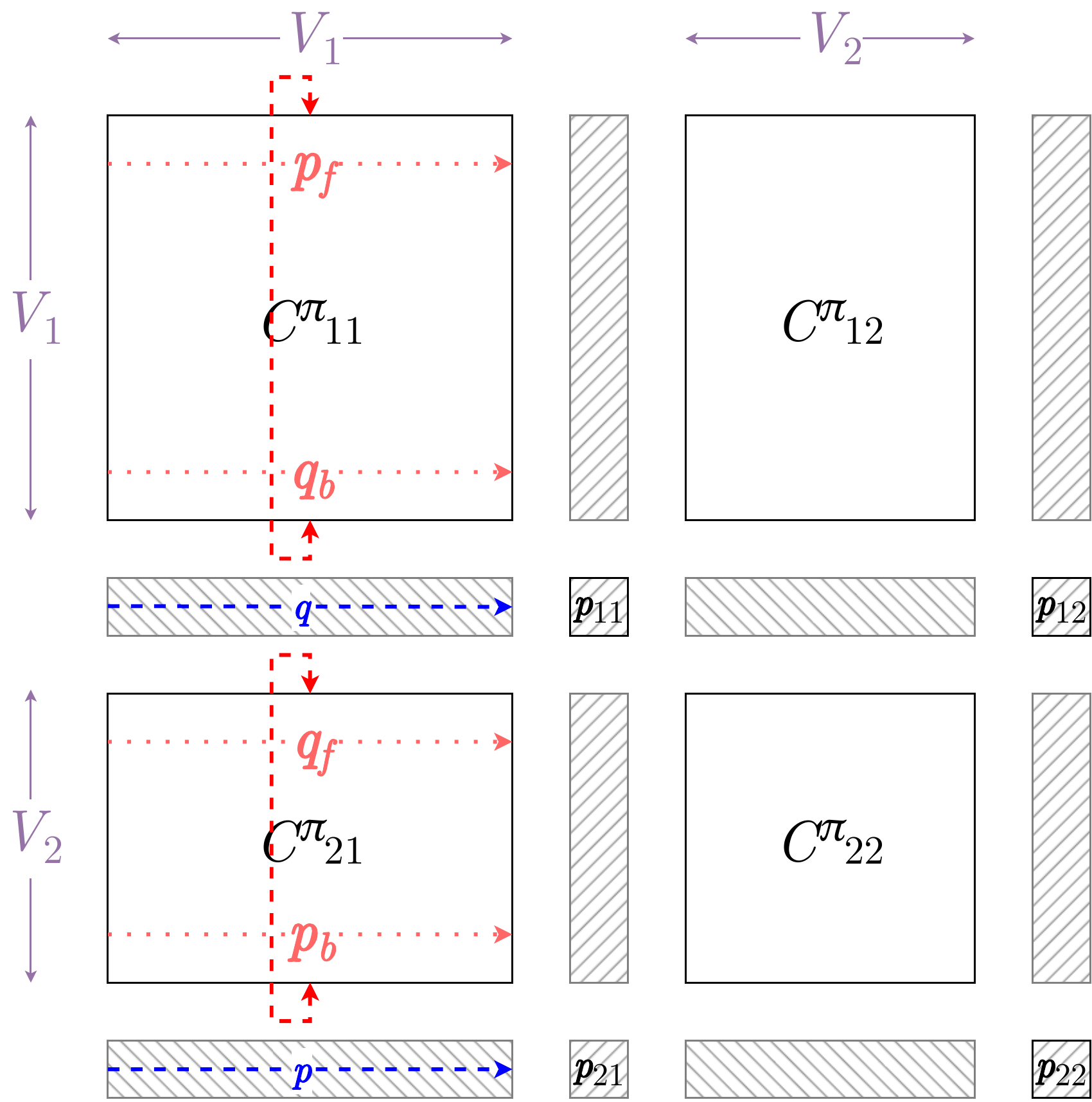}
    \caption{Vertical View (some notations in the figure are defined in appendices)}
    \label{fig:2dcutflatten}
  \end{subfigure}\hfill
  \begin{subfigure}[b]{0.60\linewidth}
    \centering
    \includegraphics[width=\linewidth]{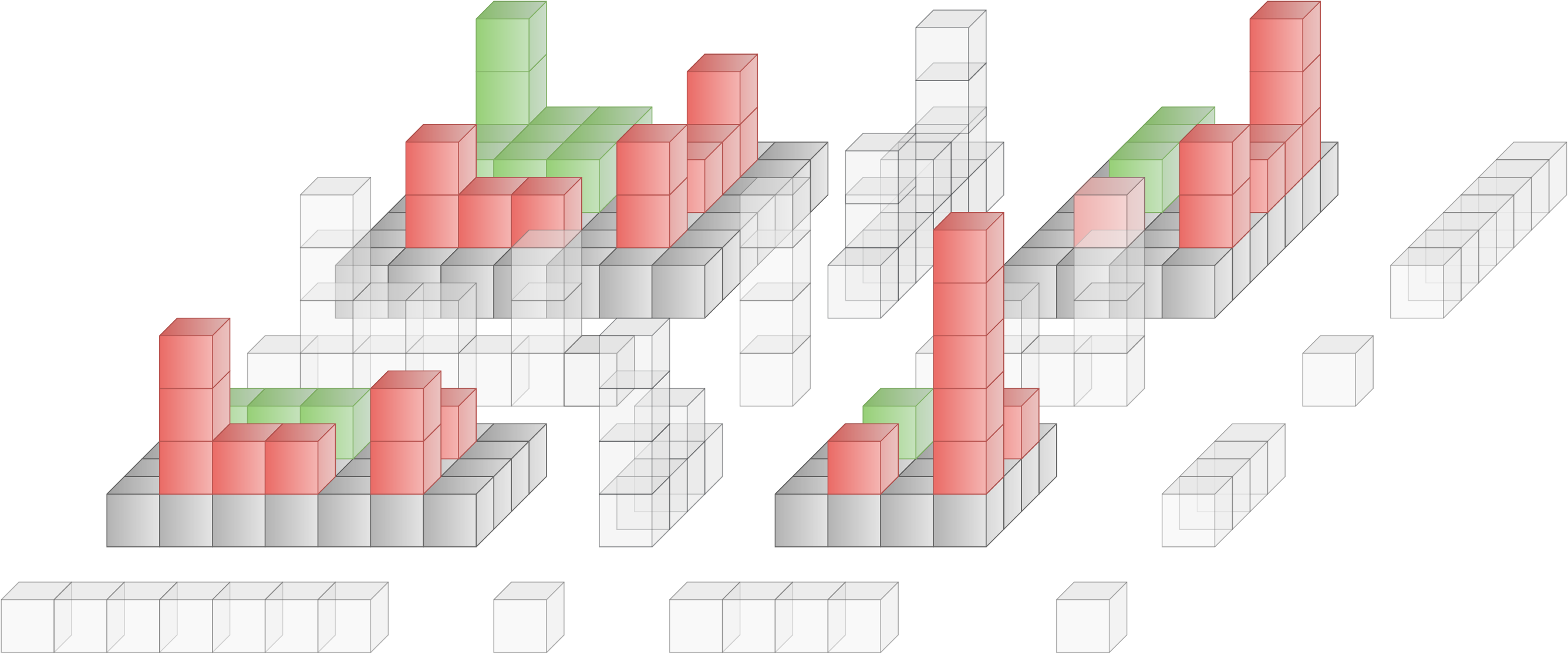}
    \caption{3D View}
    \label{fig:3dcut}
  \end{subfigure}
  \caption{Illustrations of a ``cut" operation}
\end{figure}

Taking an undirected edge $(i, j)$ that needs to be cut as the input, the algorithm can be crudely described as follows:
\begin{enumerate}
    \item  Find the two rows and two columns representing the directed edges $(i, j)$ and $(j, i)$ in $E_D$,
    i.e. the transparent nodes in Figure~\ref{fig:3dcut}
    and the shaded nodes in Figure~\ref{fig:2dcutflatten}.
    \item Push down the \textit{tag} attribute of all the nodes alongside the rows and columns,
    i.e. the red nodes and transparent nodes in Figure~\ref{fig:3dcut}.
    A node $x$ needs to be pushed down,  if some changes happen to the  nodes dominated by $x$.
    \item Cut the rows and columns, warping up the forward links and backward links of points alongside,
    as illustrated in Figure~\ref{fig:2dcutflatten}.
    This operation cuts the original Skip Orthogonal List into four smaller lists.
    \item Update the \textit{min} attribute of the remaining nodes whose \textit{tag} attribute was pushed down in step 2, i.e. the red nodes in Figure~\ref{fig:3dcut}.
    \item Return the four smaller lists that were cut out in step 3.
\end{enumerate}

The generalized
lazy propagation
to our 2D data structure ensures that only the nodes that are ``close" to the two rows and columns are  modified, and consequently the updating time is guaranteed to be low.
Specifically,
the expected time complexity is $O(\abs{V})$ upon each copy of procedure \textsc{Cut}.

\section{Our Dynamic Network Simplex Method}
\label{sec:simplex}
The simplex method performs simplex iterations on some initial feasible basis until the optimal solution is obtained.
The simplex iterations are used for refining the current solution under the dynamic changes.
In each simplex iteration,
some variable with negative simplex multiplier is selected for a copy of procedure \textsc{Pivot}, where 
one common strategy is to pivot in the variable with
the smallest simplex multiplier.
In Section~\ref{sec:dynamicmodel} we focus on defining the dynamic optimal transport operations and using simplex iterations 
to solve this problem,
while in Section~\ref{sec:algorithm} we analyze the details in each simplex iteration.
Our method is presented in the context of the conventional Network Simplex algorithm~\cite{cunningham1976network, orlin1997polynomial}.

\subsection{Dynamic Optimal Transport Operations}
\label{sec:dynamicmodel}
In an Optimal Transport problem, suppose the nodes in node set $V$ are located in some metric space $\mathcal{X}$, e.g., the Euclidean Space $\mathbb R^k$. The edge cost $c_{uv}$ is usually defined as the (squared) distance between $u$ and $v$ in the space. 
Let $w\in\mathbb R^V$ denote the weight vector as defined in the equation \eqref{eqn:redefpointweight}.
A \textbf{Dynamic Optimal Transport} algorithm should support the following four types of update as well as  online query:
\begin{itemize}
    \item \textbf{Spatial Position Modification.} Select some supply or demand point $v\in V$ and move $v$ to another point
    $v'\in \mathcal X$.
    This update usually results in the modification on an entire row or column in the cost matrix $C \in \mathbb R^{A\times B}$.
    \item \textbf{Weight Modification.} Select a pair of supply or demand points $u, v\in V$ with some positive number 
    $\delta\in\mathbb R_+$.
    Then update $w_u\gets w_u - \delta$ and $w_v\gets w_v + \delta$.
    \item \textbf{Point Deletion.} Delete
    a point $v\in V$ with $w_v = 0$ (before performing deletion, its weight should be already modified to be $0$ via the above ``weight modification'', due to the requirement of weight balance for OT).
    
    \item \textbf{Point Insertion.} Select a point $v\notin V$; 
    let $w_v = 0$ and insert $v$ into set $V$ (after the insertion, we can modify its weight from $0$ 
 to a specified value via the ``weight modification''). 
    \item \textbf{Query.} Answer the current Optimal Transport plan and the cost.
\end{itemize}

These updates do not change the overall weights in the  supply and demand sets,
and thus $\sum_{v\in V}w_v\equiv 0$ and a feasible transport
plan always exists.
Therefore we can reduce these updates to the operations on simplex basis, and we explain the ideas below:

\begin{itemize}
\item \textbf{Spatial Position Modification.} The original optimal solution is primal feasible but not primal optimal, i.e. not dual feasible.
We perform the primal simplex method based on the original optimal solution.
When moving a point $v$,
we first update the cost matrix $C$, the dual variable $\pi$ and the modified cost $C^\pi$ to meet the  constraint~\eqref{eqn:complementaryslackness}.
After that, we repeatedly perform the simplex iterations
as long as
the minimum value of the adjusted cost $C^\pi$ is negative.

\item \textbf{Weight Modification.} The original optimal solution is dual feasible but not primal feasible.
We perform the dual simplex iterations based on the original optimal solution.
Suppose we attempt to decrease $w_u$ and increase $w_v$ by $\delta$. To implement this, 
we send $\delta$ amount of flow from $u$ to $v$ in the residual network by the similar manner of the shortest path augmenting method~\cite{edmonds1972theoretical}. Specifically, 
we send the flow through basic variables.
If some variable needs to be pivoted out before the required amount of flow is sent,
we pivot in the variable with the smallest adjusted cost, and repeat this process.

\item \textbf{Point Deletion \& Point Insertion.} As the deleted/inserted point has weight $0$ (even if the weight is non-zero, we can first perform the ``weight modification'' to modify it to be zero), whether
inserting or deleting the point does not influence
our result. 
We maintain a node pool keeping  
all the supply and demand nodes with 0 weight. Each
\textbf{Point Insertion} operation takes some point from
this pool and move it to the correct spatial location (i.e., insert a new point), 
while each \textbf{Point Deletion} operation returns
a node to the pool.
\end{itemize}

Our solution updates the optimal transport
plan as soon as an update happens,
so we can answer the query for the optimal transport plan and value online.
If the number of modified nodes is not large, intuitively the optimal transport plan
should not change much, and thus we only need to run a small number of simplex
iterations to obtain the OT solution.
Assume we need to run $s$ simplex iterations, where we assume
$s\ll \abs{V}$. Then the time complexity of our algorithm
is $O(s\cdot {\rm Time}_{\rm iter})$ with ${\rm Time_{iter}}$ being the time of each simplex iteration.

\subsection{The Details for Simplex Iteration}
\label{sec:algorithm}

As discussed in
Section~\ref{sec:dynamicmodel},
the dynamic operations on OT can be effectively reduced
to simplex iterations.
In this section, we review  the operations used in
the conventional network simplex algorithm, and show how to use the data structure designed in Section~\ref{sec:datastructure} for maintaining $C^\pi$. 
The conventional network simplex method relies on the simplex method
simplified by some graph properties.
A (network) simplex iteration contains the following steps:
\begin{enumerate}
    \item \textbf{Select Variable to Pivot in.} Select the variables with the smallest adjusted cost $C^\pi$ to pivot in.
    Denote by $x_{i_{\rm in}j_{\rm in}}$ the selected one to be pivoted in.
    \item \textbf{Update Primal Solution.}
    Adding the new variable to the current basis forms a cycle. We send the circular flow in the cycle, until some basic variable in the reverse direction runs out of flow, , through Graph Search (e.g. Depth First Search) or Link/Cut Tree~\cite{sleator1981data}.
    Denote that node as $x_{i_{\rm out}j_{\rm out}}$, which is  to be pivoted out.
    \item \textbf{Update Dual Solution.} Update the dual variables $\pi$ and modified cost $C^\pi$ to meet the constraint \eqref{eqn:complementaryslackness},
    as the new basis, because $C^\pi$ will soon be queried in the next simplex iteration.
\end{enumerate}

The selecting step performs a query on the data structure on $C^\pi$ for the minimum element,
and the dual updating performs an update on the data structure.
Though the primal updating step can be done within time $O(\log\abs{V})$~\cite{tarjan1997dynamic},
the conventional network simplex maintains $C^\pi$ through brute force. That is, 
the conventional network simplex brutally traverses through all the adjusted costs and selects the minimum,
and updates the adjusted cost one by one after the dual solution is updated.
This indicates that the time complexity of each simplex iteration is $O(\abs{V}^2)$.
Our goal is to reduce this complexity; in particular, we aim to maintain $C^\pi$ so that it can answer the global minimum query and perform update when the primal basis changes.
In the simplex method, when we decide to pivot in the variable $x_{i_{\rm in}j_{\rm in}}$,
we update the dual variables as the following equation \eqref{eqn:updatedual},
\begin{equation}
\label{eqn:updatedual}
    \pi_u'\gets \pi_u + \begin{cases}c_{i_{\rm in}j_{\rm in}}^\pi & u\in V_1\\0 & u\in V_2\end{cases}
\end{equation}
where $V_1$ is the set of nodes connected to $i_{\rm in}$ and $V_2$ is the set of nodes connected to $j_{\rm in}$ after the edge $x_{i_{\rm out}j_{\rm out}}$ is cut;
$\pi$ and $\pi'$ are respectively the dual solution before and after pivoting, where $\pi_u$ and $\pi'_u$ are the entries corresponding to the node
 $u$, for each $u \in V$.

Based on the definition of the adjusted cost matrix $C^\pi$,
 the update objective can be formulated as below:
 \begin{equation*}
c^{\pi'}_{uv} = c_{uv}^\pi + \pi'_u - \pi'_v= c^\pi_{uv} + \begin{cases}
0 & (u, v)\in V_1\times V_1\\
-c^\pi_{i_{\rm in}j_{\rm in}} & (u, v)\in V_1\times V_2\\
c^\pi_{i_{\rm in}j_{\rm in}} & (u, v)\in V_2\times V_1\\
0 & (u, v)\in V_2\times V_2
\end{cases}
\end{equation*}
where $C^\pi$ is the adjusted cost matrix with regard to the old dual variables $\pi$ while $C^{\pi'}$ regards the new dual variables $\pi'$.
We present the details for updating the adjusted cost matrix in
Algorithm~\ref{alg:postpivotingupdate}.
Our Skip Orthogonal List presented in Section~\ref{sec:datastructure} is capable of performing the operations \textit{cut}, \textit{add} and \textit{link} in $O(\abs{V})$ time.
Therefore we have the following Theorem \ref{thm:main}.

\begin{theorem}
    \label{thm:main}
    Each simplex iteration in the conventional network simplex can be completed within expected $O(\abs{V})$ time.
\end{theorem}

\begin{algorithm}[htbp]
\caption{Adjusted Cost Matrix Update}
\label{alg:postpivotingupdate}
\begin{algorithmic}[1]
\Require{Adjusted cost matrix $C^\pi$ corresponding to old solution $\pi$, entering variable $x_{i_{\rm in}j_{\rm in}}$, leaving variable $x_{i_{\rm out}j_{\rm out}}$}
\Ensure{Updated adjusted cost matrix $C^{\pi'}$ corresponding to new dual solution $\pi'$}
\State $t\gets c^\pi_{i_{\rm in}j_{\rm in}}$.
\State Cut the edge $\qty(i_{\rm out}, j_{\rm out})$ in $C^\pi$ and denote the 4 pieces as $C^\pi_{11}, C^\pi_{12}, C^\pi_{21}, C^\pi_{22}$.
\State Range update on $C^\pi_{12}$ with increasing all the entries by $-t$.
\State Range update on $C^\pi_{21}$ with increasing all the entries by $t$.
\State Link the pieces $\{C^\pi_{11}, C^\pi_{12}, C^\pi_{21}, C^\pi_{22}\}$ by the edge $x_{i_{\rm in}j_{\rm in}}$,  and obtain the adjusted $C^{\pi'}$ as the output.
\end{algorithmic}
\end{algorithm}

\section{Experiments}

All the experimental results are obtained on a server equipped
with \textit{512GB} main memory of frequency \textit{3200 MHz};
the data structures are implemented in \textit{C++20}
and compiled by \textit{G++ 13.1.0} on \textit{Ubuntu 22.04.3}.
The data structures are compiled to shared objects by \textit{PyBind11}~\cite{pybind11} to be called by \textit{Python 3.11.5}.
Our code uses the Network Simplex library from~\cite{BPPH11} to obtain an initial feasible flow.

In our experiment, we use the \textbf{Network Simplex} algorithm~\cite{orlin1997polynomial} and \textbf{Sinkhorn} algorithm~\cite{cuturi2013sinkhorn} from the
Python Optimal Transport (POT) library~\cite{flamary2021pot}.
We test our algorithm for both the \textbf{Spatial Position Modification} and \textbf{Point Insertion} scenarios as described in Section~\ref{sec:dynamicmodel}.
We take running time to measure their performances.


\textbf{Datasets.} We study the performance of our algorithm on both synthetic  and real datasets.
For synthetic datasets,
we construct a mixture of two Gaussian distributions in $\mathbb R^{784}$,
where the points of the same Gaussian distribution share the same label.
We also use the popular real-world dataset MNIST~\cite{lecun2010mnist}.
We partition the labels into two groups and compute the optimal transport between them.


\textbf{Setup.} We set the Sinkhorn regularization parameter $\lambda$ as $0.1$ and scale the median of the cost matrix $C$ to be 1.
We vary the the node size $\abs{V}$ up to $4\times 10^4$.
For each dataset, we test the static running time of POT on our machine and executed each dynamic operations 100 times to calculate the means of our algorithm.
For \textit{Spatial Position Modification},
we randomly choose some point and add a random Gaussian noise with a variance of 0.5 in each dimension to it.
For \textit{Point Insertion},
we randomly select a point in the dataset that is outside the current OT instance to insert.
We perform a query after these updates and compare with the static algorithms implemented in POT.

\textbf{Result and Analysis.} We illustrate our experimental results in Figure~\ref{fig:result}.
In the dynamic scenarios,
our algorithm is about 1000 times faster than static Network Simplex algorithm and 10 times faster than Sinkhorn Algorithm
when the size $\abs{V}$ reaches 40000,
and reveals stable performance in practice.
As the number of nodes grow larger, the advantage of our algorithm becomes more significant.
This indicates that our algorithm is fast in the case when the number of simplex iterations is not large, as discussed in
Section~\ref{sec:dynamicmodel}.
Also, the running time of our algorithm is slightly higher than linear trend.
Though each simplex iteration in our method is strictly linear in expectation theoretically,
this could be influenced by several practical factors, such as the increment in number of simplex iterations, or the decrement in cache hit rate as the node size grows larger.
\begin{figure}[htbp]
  \centering
  \includegraphics[width = \linewidth]{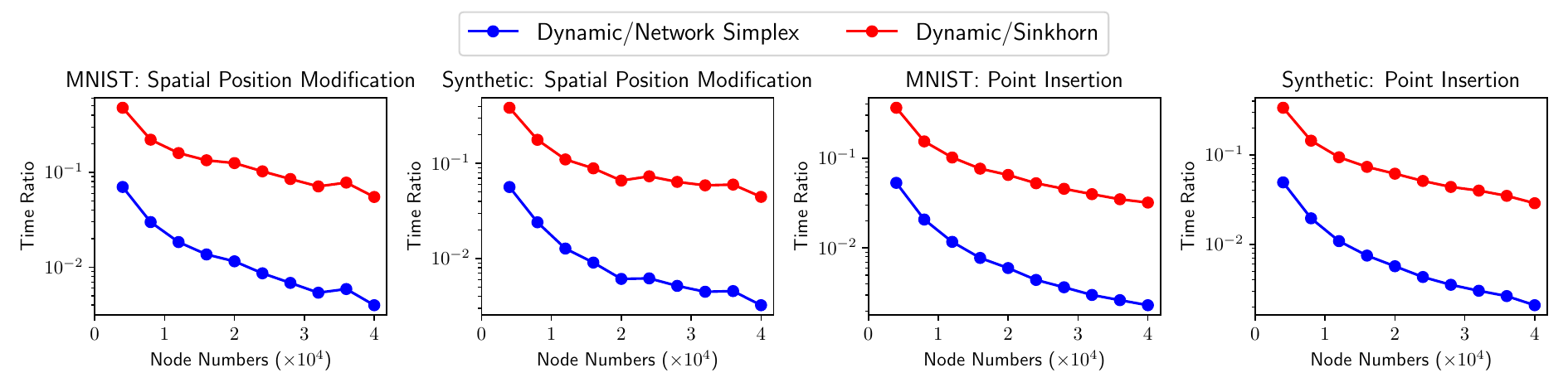}
  \caption{The ratio of the execution time of our dynamic algorithm to that of the static algorithms.}
  \label{fig:result}
\end{figure}

\section{Conclusion and Future Work}
In this paper, we propose a dynamic data structure for the  traditional network simplex method.
With the help of our data structure,
the time complexity of the whole pivoting process
is $O(\abs{V})$ in expectation. 
However, our algorithm lead to several performance issues in practice.
First, as our algorithm stores the entire 2D Skip Orthogonal List data structure,
it may take relatively high space complexity.
Second, as our algorithm is based on linked data structures,
the cache hit rate is not high.
An interesting future work for improving our implementation is to develop new algorithms and data structures with similar complexity but being more memory friendly.

\begin{appendices}
\renewcommand{\thesection}{\Alph{section}}%
\renewcommand{\thesubsection}{\thesection.\arabic{subsection}}

\section{Space Complexity of Skip Orthogonal List}
\label{sec:spaceproof}

In an Optimal Transport problem on point set $V$, the Skip Orthogonal List in our algorithm has a size of $3\abs{V} - 2 = O(\abs{V})$ rows and columns. To show that the space complexity of our Skip Orthogonal List is quadratic with high probability, we only need to prove theorem \ref{thm:space}.

\begin{theorem}
    \label{thm:space}
    The space complexity of a $n\times n$ 2D Skip Orthogonal List with parameter $p$ is $O(n^2)$ with high probability.
\end{theorem}

\begin{proof}
    Denote the number of nodes in the entire Skip Orthogonal List as $N$.
    
    Assume there are $\ell$ levels in the Skip Orthogonal List. Let $m_i$ be the number of nodes on the $i$-th level where $1\leq i\leq \ell$, and $0$ when $i > \ell$.
    It is easy to see that
    \begin{itemize}
        \item $m_1 = n$;
        \item $m_i\geq m_{i + 1}$ for all $i\geq 1$, i.e., $\qty{m_i}_{i = 1}^\infty$ is monotonically decreasing;
        \item There are $N = \sum_{i = 1}^\infty m_i^2$ nodes in the entire Skip Orthogonal List.
    \end{itemize}

    If $m_i\geq \frac 4{(1 - p)^2}\ln n$, we call the $i$-th level a \textit{big level}; otherwise we call the $i$-th level a \textit{small level}.
    By the monotonicity of $\qty{m_i}_{i = 1}^\infty$, we can see that, there exists a threshold $\ell_t$, such that $i\leq \ell_t$ if and only if $i$-th level is a big level, i.e., $i\leq \ell_t$ if and only if $m_i\geq 4{(1 - p)^2}\ln n$.

    Denote $N_b \triangleq \sum_{i = 1}^{\ell_t}m_i^2$, i.e., the number of nodes in big levels; denote $N_s\triangleq \sum_{i = \ell_t + 1}^\infty m_i^2$, i.e., the number of nodes in small levels.
    Therefore $N = N_b + N_s$. We now give their bound respectively with Lemma \ref{thm:bigbound} and Lemma \ref{thm:smallbound}.

    \begin{lemma}
        \label{thm:bigbound}
        $$
        \Pr[N_b> \frac{4}{3 - 2p - p^2}n^2]\leq \frac{\log_{2/(1 + p)}n}{n^2} = O\qty(\frac 1n)
        $$
    \end{lemma}

    \begin{proof}
        First we prove that
        
        $$
        \Pr[m_{i + 1} > \frac{1 + p}2 m_i\ \bigg|\ m_i \geq\frac{4}{(1 - p)^2}\ln n] \leq \frac 1{n^2}
        $$

        Denote $\qty{X_j}_{j = 1}^{m_i}$ be the indicator variable of $m_i$ events, such that $X_j = 1$ if and only if the $j$-th row and column of the $i$-th level exists in $i + 1$-th level.
        By the construction of Skip Orthogonal List, we can see that $\qty{X_j}_{j = 1}^{m_i}$ are independent Bernoulli variables with parameter $p$.
        Also, by definition, $m_{i + 1} = \sum_{j = 1}^{m_i}X_j$
        Therefore
        \begin{itemize}
            \item By linearity of expectation
            $$
            \mathbb E\qty[m_{i + 1}] = \sum_{j = 1}^{m_i}\mathbb E\qty[X_j] = pm_i
            $$
            \item By Hoeffding's inequality
            \begin{align*}
            \Pr[m_{i + 1} > \frac{1 + p}2m_i]&\leq \exp(-\frac{2\qty(\frac{1 + p}2m_i - \mathbb E\qty[m_{i + 1}])^2}{m_i}) = \exp(-\frac{(1 - p)^2}2m_i)
            \end{align*}
        \end{itemize}
        Therefore,
        \begin{align*}
        \Pr[m_{i + 1} > \frac{1 + p}2m_i\ \bigg|\ m_i\geq \frac{4}{(1 - p)^2}\ln n]\leq \exp(-\frac{(1 - p)^2}2\cdot \frac 4{(1 - p)^2}\ln n) = \frac 1{n^2}
        \end{align*}


        When $\frac {1 + p}2m_i\geq m_{i + 1}$ for all $1\leq i\leq \ell_t - 1$,

        $$
        N_b\sum_{i = 1}^{\ell_t}m_i^2\leq \sum_{i = 0}^\infty m_1^2\qty(\frac{1 + p}2)^{2i} = \frac4{3 - 2p - p^2}n^2
        $$

        Therefore, when $N_b > \frac 4{3 - 2p - p^2}n^2$,
        there must exists some $i$ in range $1\leq i\leq \min(\log_{2/(1 + p)}n, \ell_t - 1)$ such that $m_{i + 1} > \frac{1 + p}2m_i$, since $m_1, m_2, \cdots, m_{\ell_t}$ are positive integers.
        By union bound, this happens with probability less than or equal to $\frac {\log_{2/(1 + p)}n}{n^2}$.
    \end{proof}
    \begin{lemma}
        \label{thm:smallbound}
        $$
        \Pr[N_s > \frac 8{\ln \frac 1p(1 - p)^2}\ln^3 n]\leq \frac {4\ln n}{(1 - p)^2n^2}
        $$
    \end{lemma}
    \begin{proof}
        In order to bound $N_s$, we separate $\qty{m_i}_{i = \ell_t + 1}^\infty$
        into $\left\lfloor\frac 4{(1 - p)^2}\ln n\right\rfloor$ parts:
        the $j$-th part contains the number of levels of which sizes are $j\times j$.
        Denote the number of levels in the $j$-th part as $k_j$.
        Thus $N_s = \sum_{i = \ell_t + 1}^\infty m_i^2 = \sum_{j = 1}^{\left\lfloor\frac 4{(1 - p)^2}\ln n\right\rfloor} j^2k_j$

        If $k_j > 0$, which indicates that there are some consecutive levels in the Skip Orthogonal List that indeed have size $j$.
        Since for all $i\geq 1$, $\Pr[m_i = m_{i + 1}] = p^{m_i}$ (all Bernoulli variables indicating whether the row/column remains in the upper level turns true), and the event $\qty{[m_i = m_{i + 1}]}_{i = 1}^\infty$ are independent, we can see that $\Pr[k_j > \frac2{j\ln\frac 1p}\ln n]\leq \frac 1{n^2}$.

        Therefore, by union bound, with probability at least $1 - \frac 4{(1 - p)^2}\ln n$ we have $k_j\leq \frac 2{j\ln \frac 1p}\ln n$ for all $1\leq j\leq \frac 4{(1 - p)^2}\ln n$ and thus
        \begin{align*}
        N_s = \sum_{j = 1}^{\left\lfloor\frac 4{(1 - p)^2}\ln n\right\rfloor} j^2k_j \leq \frac2{\ln\frac 1p}\ln n\sum_{j = 1}^{\left\lfloor\frac 4{(1 - p)^2}\ln n\right\rfloor}j \leq \frac 8{\ln \frac 1p(1 - p)^2}\ln^3 n
        \end{align*}
        which is equivalent to the description of Lemma \ref{thm:smallbound}.
    \end{proof}

    By Lemma \ref{thm:bigbound}, \ref{thm:smallbound} and union bound,
    as $p\in (0, 1)$ is a constant, we can see that there exists some constant $c_1, c_2$ such that $N < c_1n^2$ with probability at least $1 - \frac{c_2}n$, which finishes the proof of Theorem \ref{thm:space}.
\end{proof}

\section{Operations on Skip Orthogonal List}
\label{sec:timeproof}

This section provides detailed algorithm for operations on our Skip Orthogonal List.
We use the notation in Figure \ref{fig:2dcutflatten} and denote $a^{(i)}$ as the $i^{\rm th}$ ancestor of node $a$.

\subsection{Cut}
\label{sec:timeproofcut}
\begin{algorithm}[htbp]
\caption{Update (Cut)}
\label{alg:cut}
\begin{algorithmic}[1]
\Require{Skip 2D Orthogonal List $C^\pi$, Edge $\qty{u, v}$ to cut}
\Ensure{Shattered $C^\pi$ after cutting the edge $\qty{u, v}$}
\Procedure{Cut}{$C^\pi, u, v$}
\State Find the 4 intersections of row $(u, v)$, row $(v, u)$ and column $(u, v)$, column $(v, u)$. Denote them as $p_{11}, p_{12}, p_{21}, p_{22}$ as Figure \ref{fig:2dcutflatten} demonstrates.
\State Run 4 copies of procedure \textsc{PushDownLines} to clear the \textit{tag} attribute of relevant nodes.
\State Run 4 copies of procedure \textsc{CutLine} to break the entire 2D Orthogonal List $C^\pi$ into 4 pieces.
\State Run 4 copies of Procedure \textsc{PullUpLines} to update the \textit{min} attributes of affected nodes.
\EndProcedure

\Procedure{CutLine}{$C^\pi, p, q$}
\For{each $x$ in the Euler Tour of the current sub-tree}
    \State Move $p, q$ in the current direction.
    \State Let $p_b, q_b$ be the backward node in the orthogonal direction and $p_f, q_f$ be the forward node.
    \State Iterative calculate $p_b^{(i)}, p_f^{(i)}, q_b^{(i)}, q_f^{(i)}$, i.e., the $i$-th ancestor.
    \For{$i = h_x, h_x - 1, \cdots, 0$}
        \State Link $p_b^{(i)}$ and $q_f^{(i)}$. Update \textit{parent} attributes.
        \State Link $q_b^{(i)}$ and $p_f^{(i)}$. Update \textit{parent} attributes.
        \State Call procedure \textsc{PushDown} to update the \textit{min} and \textit{tag} attribute of their children.
    \EndFor
\EndFor
\EndProcedure

\Procedure{PullUpLines}{$C^\pi, p$}
\For{each Direction \textit{dir}}
    \For{each edge $x$ in the Euler Tour of the current direction}
        \For{$i = 1, 2, \cdots, h_x$}
            \State Let $p^{(i)}$ be the $i^{\rm th}$ ancestor of $p$.
            \State Call procedure \textsc{PullUp} to update the \textit{min} attribute of $p^{(i)}$ by brutally traversing through its children set.
        \EndFor
        \State Let $p$ be the next node in the Euler Tour.
    \EndFor
\EndFor
\EndProcedure
\end{algorithmic}
\end{algorithm}

Algorithm \ref{alg:cut} describes the cut update to update the structure while updating some aggregate information
(i.e. \textit{tag} and \textit{min} attributes).
Here are some clarifications:

\begin{itemize}
    \item The time complexity of finding these rows and columns depends on the lookup table it depends on. It can never exceed $O(\abs{T}) = O(\abs{V})$.
    \item As Figure \ref{fig:2dcutflatten} demonstrates,
    on the bottom level,
    procedure \textsc{CutLine} removes 2 rows related to the edges to cut and links nodes alongside (changes the \textit{forward} and \textit{backward} pointers of nodes along side),
    with a time complexity of $\Theta(\abs{T}) = \Theta(\abs{V})$,
    while in the upper levels,
    procedure \textsc{CutLine} removes all transparent nodes in Figure \ref{fig:3dcut} and changes the \textit{forward} pointers and related \textit{backward} pointers of the nodes alongside, i.e. red nodes in Figure \ref{fig:3dcut}.

    Note that when linking adjacent nodes, we need to update the \textit{parent} attributes of nodes $q_f^{(i)}$ and $p_f^{(i)}$ accordingly, and nodes below it accordingly.
    However, every node is affected at most once during the entire \textsc{CutLine} procedure, and the number of nodes affected are sure to be affected by \textsc{PullUpLines} process.
    \item Procedure \textsc{PushDownLines} is capable of pushing down the \textit{tag} attribute of all nodes alongside the rows and columns of the four nodes,
    i.e. the red nodes in \ref{fig:3dcut}.
    \item Similar to \textsc{PushDown}, procedure \textsc{PullUpLines} updates the \textit{min} attribute of the red nodes in \ref{fig:3dcut}.
    \item In procedure \textsc{PushDown},
    $q$ is in the children set of $p$ if and only if $p$ is the parent of $q$.
    Later we will use this concept to analyze the time complexity.
    \item Procedure \textsc{PushDownLines} is similar to procedure \textsc{PullUpLines}.
    Instead of calling procedure \textsc{PushDown} which updates parents' \textit{min} attribute from lower level to higher level,
    \textsc{PullUpLines} calls procedure \textsc{PullUp} which clears the \textit{tag} attribute from higher level to lower level.
    It pushes down the row and column of the input node and the row and column next to the given node.
\end{itemize}
From analyses above,
we can now show Theorem \ref{thm:cut}.
\begin{theorem}
    \label{thm:cut}
    The expected time complexity of procedure \textsc{Cut} is $O(\abs{V})$.
\end{theorem}

\begin{proof}
    From Algorithm \ref{alg:cut},
    each \textsc{PushDownLines},
    \textsc{CutLine},
    procedure \textsc{PullUpLines}
    only apply constant number of modifications and queries on children of nodes alongside the 2 cut rows and columns
    , i.e. the red and transparent nodes and their children in Figure \ref{fig:3dcut}.\\
    We first calculate the expected number of nodes removed in procedure \textsc{Cut},
    i.e. the transparent nodes in Figure \ref{fig:3dcut}.
    The removed nodes in each level are in the row $\rho_{uv}$, $\rho_{vu}$ and column $\tau_{uv}, \tau_{vu}$ where $\rho_{uv} \triangleq \qty{\bigl((u, v), e\bigr) : e\in E(D_T)}$ and $\tau_{uv}\triangleq \qty{\bigl(e, (u, v)\bigr) : e\in E(D_T)}$.
    By the definition of node height,
    node $\bigl((u, v), e\bigr)$, $\bigl(e, (u, v)\bigr)$, $\bigl((v, u), e\bigr)$ and $\bigl(e, (v, u)\bigr)$ are in less than or equal to $h_e + 1$ orthogonal lists.
    Therefore, the total number of removed points is less than or equal to $4\sum_{e\in E(D_T)}(h_e + 1)$.
    Thus the expected number of removed nodes is less than or equal to $4\sum_{e\in E(D_T)}(\mathbb E(h_e) + 1) = \frac 4pn = O(n)$.\\    
    Next we calculate the children of nodes whose the dominate set is affected after procedure \textsc{Cut} update takes place,
    i.e. red nodes and their children after the \textsc{Cut} update in Figure \ref{fig:3dcut}.
    Similar to counting transparent nodes,
    the expected number of red nodes is also less than or equal to $\frac 4pn = O(n)$.
    Now we count the number of their children.\\
    Recall that in Figure \ref{fig:2dcutflatten} we denote the 2 diagonal pieces of the big skip list as $C^\pi_{11}$ and $C^\pi_{22}$ while we denote the 2 non-diagonal pieces of the big skip list as $C^\pi_{21}$ and $C^\pi_{12}$.
    We count the number of red nodes and their children in the 4 pieces after procedure \textsc{Cut} respectively.
    Denote $g(n)$ as the number of nodes in $C^\pi_{11}$ whose parent's children set is affected,
    i.e. whose parent is a red node as Figure \ref{fig:3dcut} demonstrates,
    when $n = \abs{E(D_{T_1})}$.
    Then $g(\abs{E(D_{T_2})})$ is that of $C^\pi_{22}$, and we know that:
    \begin{itemize}
        \item With probability $(1 - p)^n$,
        the height of all $n$ rows and columns are $0$,
        which indicates that every node in the current skip orthogonal list does not have a parent.
        Further more, as shown in Figure \ref{fig:toplevel},
        where
        there is only one row and one column beside the row just cut
        (the red row and column),
        therefore $g(n) = 2n - 1$.
        \item For $k = 1, \cdots, n$,
        $h_n = h_{n - 1} = \cdots = h_{n - k + 2} = 0$ and $h_{n - k + 1} > 0$ occurs with probability $p^{k - 1}(1 - p)$,
        as shown in Figure \ref{fig:nottoplevel},
        where the yellow nodes are the position of the cut-affected node in the upper layer and the green nodes denote upper layer nodes whose children set is not affected.
        Therefore,
        $n^2 - (n - k)^2\leq 2kn$ nodes of the current level is a child of some cut-affected node.
    \end{itemize}
    \begin{figure}[htbp]
      \centering
      \begin{subfigure}[b]{0.45\linewidth}
        \centering
        \adjustbox{width=\textwidth,height=6cm,keepaspectratio}{
          \includegraphics{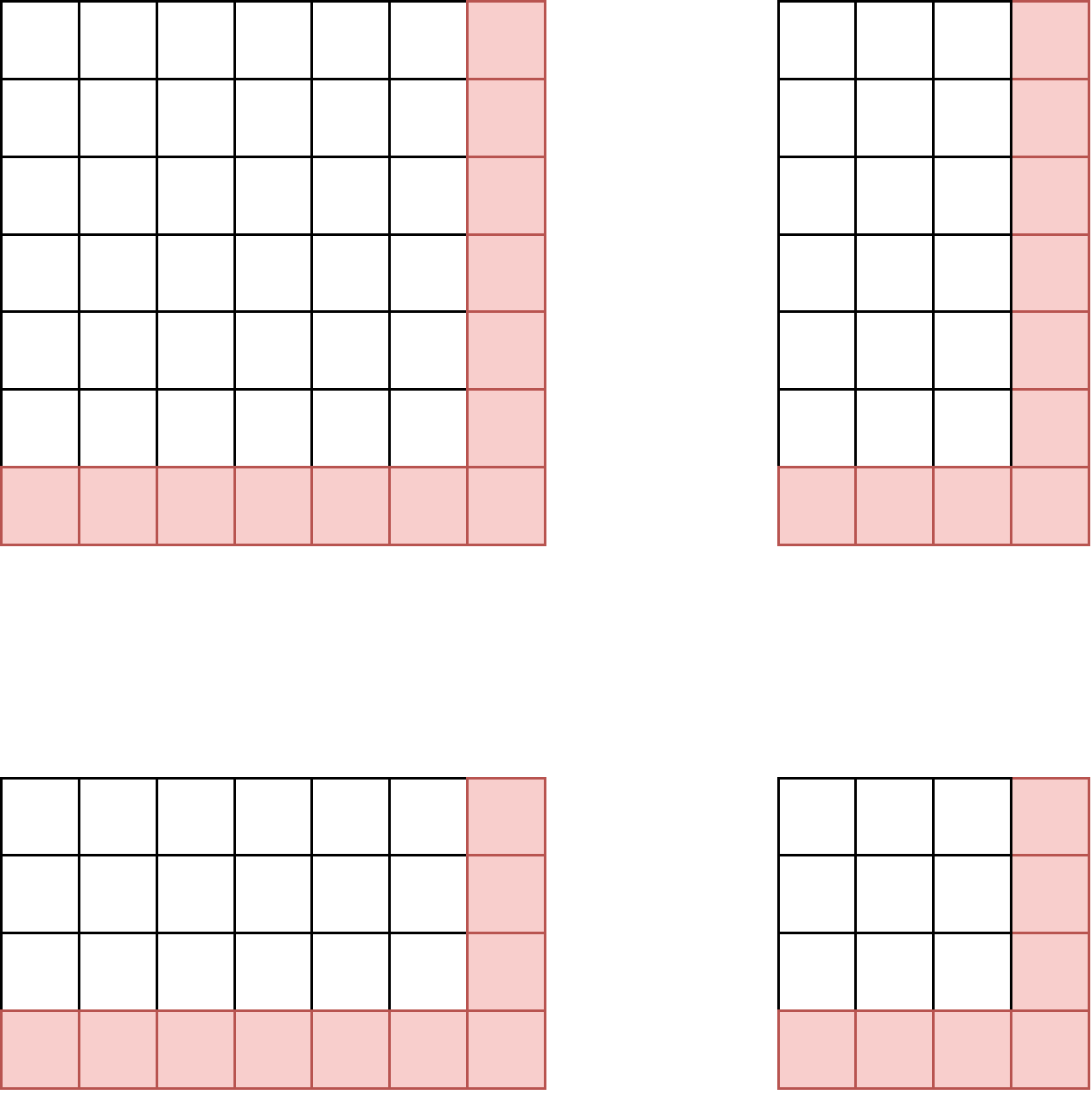}
        }
        \caption{Current Layer is at Top}
        \label{fig:toplevel}
      \end{subfigure}\hfill
      \begin{subfigure}[b]{0.45\linewidth}
        \centering
        \adjustbox{width=\textwidth,height=6cm,keepaspectratio}{
          \includegraphics{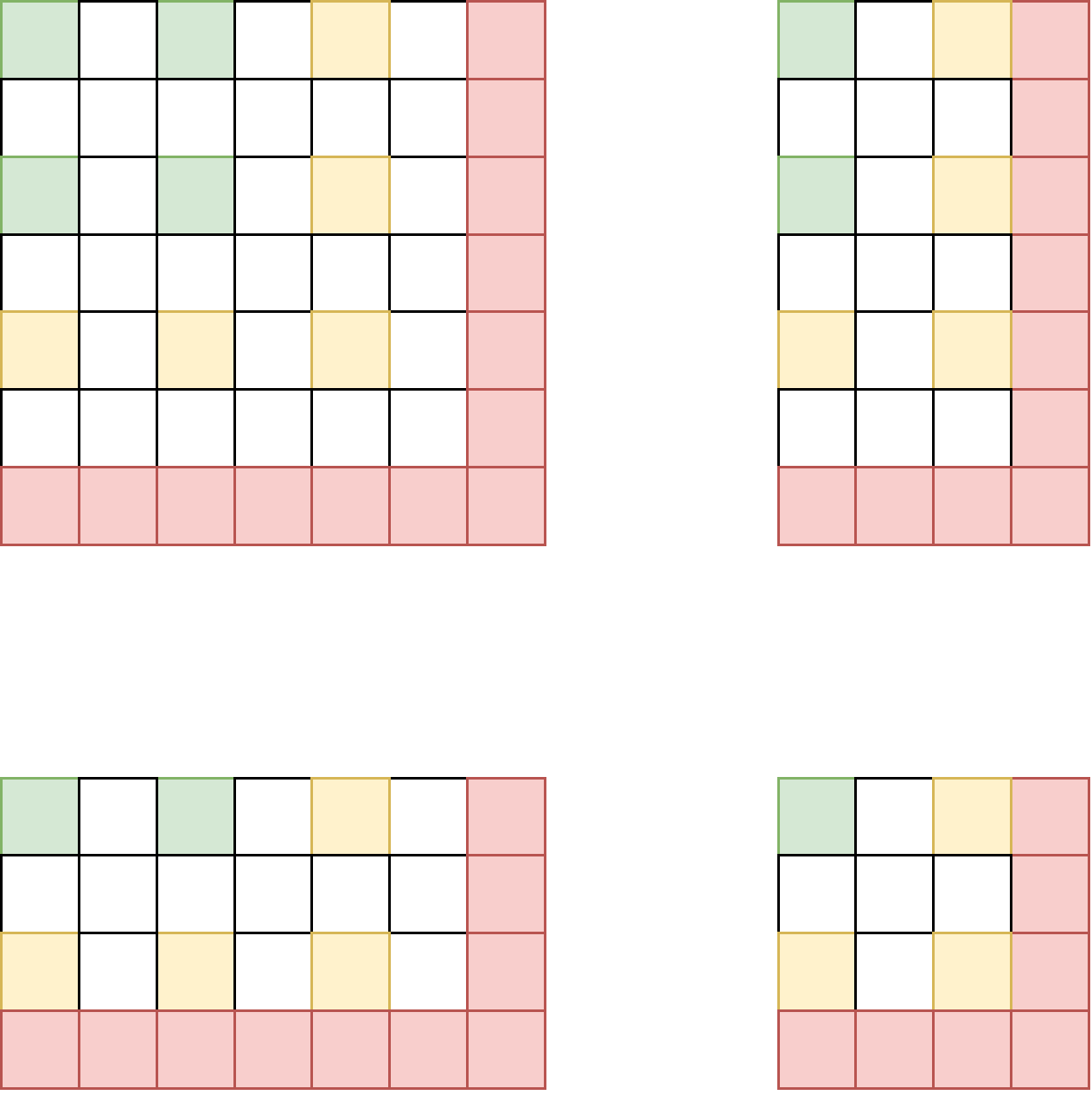}
        }
        \caption{Current Layer is Not at Top}
        \label{fig:nottoplevel}
      \end{subfigure}
      \caption{Illustration of Modified Nodes at Some Layer}
    \end{figure} 
    For $k = 0, 1, \cdots, n$,
    with probability $p^k(1 - p)^{n - k}\binom nk$,
    the upper level contains $k$ nodes,
    i.e. contain $g(k)$ nodes which is visited in procedure \textsc{Cut}
    (denote $g(0) = 0$).
    By the linearity of expectation, the following holds:

    \begin{align*}
    \mathbb E[g(n)] &= (1 - p)^n(2n - 1) + \sum_{k = 1}^np^{k - 1}(1 - p)(n^2 - (n - k)^2) + \sum_{k = 0}^n\binom nkp^k(1 - p)^{n - k} \mathbb E[g(k)]\\
    &\leq 2n + \sum_{k = 1}^\infty p^{k - 1}(1 - p)2nk + \sum_{k = 0}^n\binom nkp^k(1 - p)^{n - k} \mathbb E[g(k)]\\
    &= \frac{4 - 2p}{1 - p}n + \sum_{k = 0}^n\binom nkp^k(1 - p)^{n - k}\mathbb E[g(k)]
    \end{align*}

    We can prove by induction that $\mathbb E[g(n)]\leq \frac {4 - 2p}{(1 - p)^2}n$.
    Therefore, the expected number of visited nodes in procedure \textsc{Cut}
    (i.e. red nodes and their children in Figure \ref{fig:3dcut})
    in the 2 diagonal pieces $C^\pi_{11}$ and $c^\pi_{12}$ does not exceed $\mathbb E[g(\abs{E(D_{T_1})})] + \mathbb E[g(\abs{E(D_{T_2})})]\leq \frac{4 - 2p}{(1 - p)^2}\abs{E(D_T)} = O(\abs{E(D_T)})$.\\
    Denote $G(n, m)$ as the number of nodes visited in the cut process in each of the 2 non-diagonal pieces $C^\pi_{12}$ and $C^\pi_{21}$
    (their numbers are equal by symmetry)
    where $\abs{E(D_{T_1})} = n$ and $\abs{E(D_{T_2})} = m$.
    Similarly, we can show that
\begin{equation*}
\begin{aligned}
\mathbb E[G(n, m)] =&\ ((1 - p)^n + (1 - p)^m - (1 - p)^{n + m})(n + m - 1) \\
& + \sum_{\substack{1\leq i\leq n\\1\leq j\leq m}}p^{i + j - 2}(1 - p)^2(nm - (n - i)(m - j)) \\
& + \sum_{\substack{0\leq i\leq n\\0\leq j\leq m}} \binom ni\binom mjp^{i + j}(1 - p)^{n + m - i - j} \times \mathbb E[G(i, j)]
\end{aligned}
\end{equation*}

    Similarly, we are able to prove by induction that $\mathbb E[G(n, m)] = O(n + m)$,
    which indicates that the number of cut-affected nodes in $C^\pi_{12}$ and $C^\pi_{21}$ is $O(\abs{E(D_{T_1})} + \abs{E(D_{T_2})}) = O(\abs{E(D_T)})$.\\
    Since $\abs{E(D_T)} = 3\abs{V} - 2 = O(\abs{V})$,
    the entire \textsc{Cut} process will visit $O(\abs{V})$ nodes.
    Since each copy of procedure \textsc{Cut} only does constant number of operations on these nodes,
    the expected time complexity of procedure \textsc{Cut} is $O(\abs{V})$.
\end{proof}

\subsection{Link}
\label{sec:timeprooflink}
Procedure \textsc{Link} is very similar to procedure \textsc{Cut}.
It can be described as Algorithm \ref{alg:link}.
\begin{algorithm}[htbp]
\caption{Update (Link)}
\label{alg:link}
\begin{algorithmic}[1]
\Require{Shattered Skip Orthogonal List $C^\pi$, Edge $\qty{u, v}$ to link}
\Ensure{$C^\pi$ after linking the edge $\qty{u, v}$ is}
\State Push down \textit{tag} attribute of nodes along row $(u, u)$, row $(v, v)$, column $(u, u)$ and column $(v, v)$ through \textsc{PushDownLines}.
\State Let $h_{uv}, h_{vu}$ be 2 random variables subject to geometric distribution.
\State Link row $(u, u)$ and row $(v, v)$; column $(u, u)$ and column $(v, v)$ with 2 rows and columns of height $h_{uv}$ and $h_{vu}$ through procedure \textsc{LinkLines}.
\State Use \textsc{PushUpLines} procedure to push up \textit{min} attributes of ancestors of nodes in row $(u, u)$, row $(u, v)$, row $(v, v)$, row $(v, u)$ and these columns.
\end{algorithmic}
\end{algorithm}

Here, procedure \textsc{LinkLine} function is very similar to procedure \textsc{CutLine}.
It creates a series of new node according to their heights $h_{uv}$ and $h_{vu}$ to link 2 Skip Orthogonal List pieces together in one direction.
As procedure \textsc{link} behaves almost the same as procedure \textsc{Cut},
making constant number of modifications to nodes along the newly linked line,
the time complexity for each copy of procedure \textsc{Link} is also $O(\abs{V})$.

\subsection{Insert}
\label{sec:timeproofinsert}
For procedure \textsc{Insert},
we only need to add one basic variable to the basis to make it feasible again.
Without loss of generality,
Assume we need to add node $v$ to the demand node set.
Procedure \textsc{Insert} could be described in Algorithm \ref{alg:insert}.

\begin{algorithm}[htbp]
\caption{Update (Insert)}
\label{alg:insert}
\begin{algorithmic}[1]
\Require{2D Skip Orthogonal List $C^\pi$, node $v$ to insert}
\Ensure{$C^\pi$ after inserting node $v$}
\State $u\gets \arg\max_{u\in V}\pi_u - C_{uv}$
\State $\pi_v\gets \pi_u - C_{uv}$
\State Insert $v$ into set $V$ and insert $x_{uv}$ into set of basic variables $X_B$.
\end{algorithmic}
\end{algorithm}

Since $\pi_v = \pi_u - C^{uv}$ and $x_{uv}$ is the newly added basic variable,
dual variable $\pi$ satisfies Constraint \ref{eqn:complementaryslackness}.
Further more,
for all $\tilde u\in V$,
we have $\pi_v = \pi_u - C_{uv} = \max_{u\in V}\pi_u - C_{uv}\geq \pi_{\tilde u} - C_{\tilde u v}$.
Therefore $C^\pi_{\tilde uv} = C_{uv} + \pi_v - \pi_u\geq 0$,
which indicates that simplex multipliers related to the newly added variable is non-negative.
If the old basis is primal optimal,
the new basis will sure be primal optimal too.\\
Line 3 in Algorithm \ref{alg:insert} requires adding an entire row and column to the Skip Orthogonal List based on the existing basic variables $X_B$.
together with a copy of procedure \textsc{Link} to update $X_B$.
The expected time complexity for these operations are all $O(\abs{V})$.
Therefore the expected time complexity of each copy of Algorithm \ref{alg:insert} is $O(\abs{V})$.

\subsection{Range Update (Add) and Query}

The \textit{tag} attribute of each node stores the value that should be added to each node dominated by it but haven't been propagated.
The \textit{min} attribute stores the minimum value and the index of the minimum value of all nodes dominated by it.
Similar to the lazy propagation technique:

\begin{itemize}
    \item For procedure \textsc{RangeUpdate} to do the range update that adds all elements by \textit{val} amount in the current Skip Orthogonal List piece, we add the \textit{tag} attribute and \textit{min} attribute of all nodes by \textit{val} on the top layer.
    \item For query \textit{min} to do range query on the minimum value and indices of all elements in the Skip Orthogonal List, we visit every top layer node and find the node with the minimum \textit{min} attribute.
\end{itemize}

Therefore, to bound the running time of procedure \textsc{RangeUpadte} and query \textsc{GlobalMinimum},
we only need to give bound to the nodes on the top level.
To achieve this, we show Lemma \ref{thm:piecetop} and \ref{thm:top}.
But in order to prove them, we need to prove Lemma \ref{thm:maxelems} first.

\begin{lemma}
    \label{thm:maxelems}
    Let $n$ be a positive integer and $p$ be a real number in range $(0, 1)$.
    For any $n$ independent and identically distributed random variables $\qty{X_i}_{i = 1}^n$ following geometric distribution with parameter $p$,
    the expected number of maximum elements is less than $\frac 1p$, and the expected square of the number of the maximum elements is less than $\frac 2{p^2} - \frac 1p$,
    i.e.,
    $$
    \mathbb E\qty[\abs{\qty{i : X_i = \max_{j = 1}^nX_j}}] < \frac 1p
    $$
    $$
    \mathbb E\qty[\abs{\qty{i : X_i = \max_{J = 1}^nX_j}}^2] < \frac 2{p^2} - \frac 1p
    $$
\end{lemma}
\begin{proof}
Let $Y_i$ be the number of maximum elements in subsequence of the first $i$ elements, i.e.,
$$
Y_i\triangleq \abs{\qty{k : X_k = \max_{j = 1}^iX_j}}
$$
Therefore, Lemma \ref{thm:maxelems} is equivalent to $\mathbb E\qty[Y_n] < \frac 1p$ and $\mathbb E\qty[Y_n] < \frac 2{p^2} - \frac 1p$.
We now prove that $\mathbb E\qty[Y_i] < \frac 1p$ and $\mathbb E\qty[Y_i^2] < \frac 2p\mathbb E\qty[Y_i] - \frac 1p$ for all integer $i = 1, 2, \cdots, n$ by induction.

First, $Y_1 = 1$, because $\qty{X_1}$ contains only 1 element.
Therefore, $\mathbb E\qty[Y_1] < \frac 1p$ and $\mathbb E\qty[Y_1] < \frac 2p\mathbb E\qty[Y_1] - \frac 1p$.
Next, we prove that $\mathbb E\qty[Y_{i + 1}] < \frac 1p$ and $\mathbb E\qty[Y_{i + 1}^2] < \frac 2p\mathbb E\qty[Y_{i + 1}] - \frac 1p$ are implied by $\mathbb E\qty[Y_i] < \frac 1p$ and $\mathbb E\qty[Y_i^2] < \frac 2p\mathbb E\qty[Y_i] - \frac 1p$.

When calculating $Y_{i + 1}$ from $Y_i$, there are 3 cases:
\begin{enumerate}
    \item $X_{i + 1} > \max_{j = 1}^iX_j$. Suppose this happens with probability $P_1$. In this case, $Y_{i + 1} = 1$.
    \item $X_{i + 1} = \max_{j = 1}^iX_j$. Suppose this happens with probability $P_2$. In this case, $Y_{i + 1} = Y_i + 1$.
    \item $X_{i + 1} < \max_{j = 1}^iX_j$. Suppose this happens with probability $P_3$. In this case, $Y_{i + 1} = Y_i$.
\end{enumerate}

Notice that, by law of total probability,
\begin{figure}
\begin{subfigure}{0.45\linewidth}
\centering
\begin{align*}
    P_1 &= \Pr[X_{i + 1} > \max_{j = 1}^iX_j] \\&= \sum_{\ell = 1}^\infty \Pr[\max_{j = 1}^iX_j = \ell]\Pr[X_{i + 1} > \ell] \\&= p\sum_{\ell = 1}^\infty \Pr[\max_{j = 1}^iX_j = \ell]p^{\ell - 1}
\end{align*}
\end{subfigure}\hfill
\begin{subfigure}{0.45\linewidth}
\centering
\begin{align*}
    P_2&= \Pr[X_{i + 1} = \max_{j = 1}^iX_j] \\&= \sum_{\ell = 1}^\infty \Pr[\max_{j = 1}^iX_j = \ell]\Pr[X_{i + 1} = \ell] \\&= (1 - p)\sum_{\ell = 1}^\infty \Pr[\max_{j = 1}^iX_j = \ell]p^{\ell - 1}
\end{align*}
\end{subfigure}
\end{figure}

Let $t\triangleq \sum_{\ell = 1}^\infty \Pr[\max_{j = 1}^iX_j = \ell]p^{\ell - 1}$. Therefore, $t$ is bounded by the interval $[0, 1]$, and we can express $P_1 = pt$, $P_2 = (1 - p)t$, $P_1 + P_2 + P_3 = 1$.

First we prove $\mathbb E\qty[Y_{i + 1}] < \frac 1p$.
By law of total expectation, we have

\begin{align*}
\mathbb E\qty[Y_{i + 1}] &= 1\cdot P_1 + (\mathbb E\qty[Y_i] + 1)\cdot P_2 + \mathbb E\qty[Y_i]\cdot P_3 = \qty(1 - \mathbb E\qty[Y_i]p)t + \mathbb E\qty[Y_i]
\end{align*}

Because $\mathbb E\qty[Y_i] < \frac 1p$ (inductive hypothesis),
$1 - \mathbb E\qty[Y_i]p > 0$.
Since $t\leq 1, 0 < p < 1$, we get the following

\begin{align*}
    \mathbb E\qty[Y_{i + 1}] = \qty(1 - \mathbb E\qty[Y_i]p)t + \mathbb E\qty[Y_i] \leq 1 - \mathbb E\qty[Y_i]p + \mathbb E\qty[Y_i] = 1 + (1 - p)\mathbb E\qty[Y_i] < \frac 1p
\end{align*}

Hence finishes the proof of $\mathbb E\qty[Y_{i + 1}] < \frac 1p$.

Next we prove that $\mathbb E\qty[Y_{i + 1}^2] < \frac 2p\mathbb E\qty[Y_{i + 1}] - \frac 1p$.
Since $Y_{i + 1} = (1 - \mathbb E\qty[Y_i]p)t + \mathbb E\qty[Y_i] \geq \mathbb E\qty[Y_i]$,
by law of total expectation, we have
\begin{align*}
\mathbb E\qty[Y_{i + 1}^2] =&\ 1\cdot P_1 + (\mathbb E\qty[Y_i] + 1)^2\cdot P_2 + \mathbb E^2[Y_i]\cdot P_3\\
=&\ \qty(-\mathbb E\qty[Y_i^2]p - 1 + 2\mathbb E\qty[Y_i])t + 2\mathbb E\qty[Y_{i + 1}] + \mathbb E\qty[Y_i^2] - 2\mathbb E\qty[Y_i]
\end{align*}

Because $\mathbb E\qty[Y_i^2] < \frac 2p\mathbb E\qty[Y_i] - \frac 1p$,
we have $-\mathbb E\qty[Y_i^2] - 1 + 2\mathbb E\qty[Y_i] > 0$.
Since $t\leq 1, 0 < p < 1$, we get the following

\begin{align*}
    \mathbb E\qty[Y_{i + 1}^2] =&\ \qty(-\mathbb E\qty[Y_i^2]p - 1 + 2\mathbb E\qty[Y_i])t + 2\mathbb E\qty[Y_{i + 1}] + \mathbb E\qty[Y_i^2] - 2\mathbb E\qty[Y_i]\\
    \leq &\ {-\mathbb E}\qty[Y_i^2]p - 1 + 2\mathbb E\qty[Y_i] + 2\mathbb E\qty[Y_{i + 1}] + \mathbb E\qty[Y_i^2] - 2\mathbb E\qty[Y_i]\\
    =&\ (1 - p)\mathbb E\qty[Y_i^2] - 1 + 2\mathbb E\qty[Y_{i + 1}]\\
    <&\ (1 - p)\qty(\frac 2p\mathbb E\qty[Y_i] - \frac 1p) - 1 + 2\mathbb E\qty[Y_{i + 1}]\\
    \leq&\ (1 - p)\qty(\frac 2p\mathbb E\qty[Y_{i + 1}] - \frac 1p) - 1 + 2\mathbb E\qty[Y_{i + 1}]\\
    =&\ \frac 2p\mathbb E\qty[Y_{i + 1}] - \frac 1p
\end{align*}

Therefore, we have finished the proof that $\mathbb E\qty[Y_n] < \frac 1p$ and $\mathbb E\qty[Y_n^2] < \frac 2p\mathbb E\qty[Y_i] - \frac 1p$, and therefore $\mathbb E\qty[Y_n^2] < \frac 2{p^2} - \frac 1p$. Thus Lemma 3 is proved.

\end{proof}

With Lemma \ref{thm:maxelems}, we can now prove Lemma \ref{thm:piecetop}, which immediate leads to the expected time complexity upper bound of procedure \textsc{RangeUpdate}.

\begin{lemma}
    \label{thm:piecetop}
    If Skip Orthogonal List $C^\pi$ of size $n\times n$ is break into 4 pieces $C^\pi_{11}, C^\pi_{12}, C^\pi_{21}, C^\pi_{22}$ through procedure \textsc{Cut},
    then the expected size of the top level of the 2 non-diagonal pieces $C^\pi_{12}$ and $C^\pi_{21}$ is $O(n)$.
\end{lemma}

\begin{proof}
    Suppose constant $p\in (0, 1)$ is the parameter of the Skip Orthogonal List.
    Suppose the 2 non-diagonal pieces have size $m\times k$ and $k\times m$ respectively.
    By the properties of procedure \textsc{Cut},
    $m + k = n - 2$.

    For the $m\times k$ piece, denote the height variable corresponding to the rows as $\qty{a_i}_{i = 1}^m$
    and the height variable corresponding to the columns as $\qty{b_j}_{j = 1}^k$.
    Therefore, $\qty{a_i}_{i = 1}^m$ and $\qty{b_j}_{j = 1}^k$ are i.i.d. following geometric distribution with parameter $p$.
    The height of the entire Skip Orthogonal List is $\min\qty(\max_{i = 1}^ma_i, \max_{j = 1}^kb_j)$.

    If node $(i', j')$, i.e., node on the $i'$-th row and $j'$-th column, is on the top level, then $\min(a_{i'}, b_{j'}) = \min\qty(\max_{i = 1}^ma_i, \max_{j = 1}^kb_j)$, which is implied by $a_{i'} = \max_{i = 1}^ma_i\vee b_{j'} = \max_{j = 1}^kb_j$, i.e., $a_{i'}$ is the maximum element of $\qty{a_i}_{i = 1}^m$ or $b_{j'}$ is the maximum element of $\qty{b_j}_{j = 1}^k$.

    Let $I_{ij}$ be the indicator variable that node $(i, j)$ is on the top level.
    Let $I^{(r)}_i$ be the indicator variable that $a_i$ is the maximum element of $\qty{a_i}_{i = 1}^m$ and $I^{(c)}_j$ be the indicator that $b_j$ is the maximum element of $\qty{b_j}_{j = 1}^k$.
    By linearity expectation, the number of elements on the top level is $\sum_{i = 1}^m\sum_{j = 1}^k\mathbb E\qty[I_{ij}]$,
    the number of maximum elements of array $\qty{a_i}_{i = 1}^m$ is $\sum_{i = 1}^m\mathbb E\qty[I^{(r)}_i]$ and the maximum elements of array $\qty{b_j}_{j = 1}^k$ is $\sum_{j = 1}^k\mathbb E\qty[I^{(c)}_j]$.
    From the discussion above, by union bound, we have

    \begin{align*}
        \sum_{i = 1}^m\sum_{j = 1}^k\mathbb E\qty[I_{ij}]
        =&\ \sum_{i = 1}^m\sum_{j = 1}^k\Pr[\min(a_i, b_j) = \min\qty(\max_{i' = 1}^ma_{i'}, \max_{j' = 1}^kb_{j'})]\\
        \leq&\ \sum_{i = 1}^m\sum_{j = 1}^k\Pr[a_i = \max_{i' = 1}^ma_{i'}\vee b_j = \max_{j' = 1}^kb_{j'}]\\
        \leq&\ \sum_{i = 1}^m\sum_{j = 1}^k\Pr[a_i = \max_{i' = 1}^ma_{i'}] + \Pr[b_j = \max_{j' = 1}^kb_{j'}]\\
        \leq&\ k\sum_{i = 1}^m\mathbb E\qty[I^{(r)}_i] + m\sum_{j = 1}^k\mathbb E\qty[I^{(c)}_j]
    \end{align*}

    As $\sum_{i = 1}^mI^{(r)}_i$ equals to the total number of maximum values in $\qty{a_i}_{i = 1}^m$, by Lemma \ref{thm:maxelems},
    we have $\sum_{i = 1}^mI^{(r)}_i < \frac 1p$.
    Similarly, $\sum_{j = 1}^k\mathbb E\qty[I^{(c)}_j] < \frac 1p$.
    As a result, the total number of variables on the top level is smaller than $k\cdot \frac 1p + m\cdot \frac 1p < \frac np = O(n)$, which finishes the proof of Lemma \ref{thm:piecetop}.
\end{proof}

Further more, with Lemma \ref{thm:maxelems}, we can directly derive Lemma \ref{thm:top}.

\begin{lemma}
    \label{thm:top}
    The expected size of the top level of a Skip Orthogonal List $C^\pi$ is $O(1)$.
\end{lemma}

Now, with Lemma \ref{thm:piecetop} and \ref{thm:top}, the time complexity for procedure \textsc{RangeUpdate} and \textsc{Query} are directly obtained.

\section{Supplementary Experiments}

Similar to the previous experiments, the experimental results in this section are obtained from the same server with the same dataset pools,
i.e. Synthetic dataset about Gaussian Distributions on $\mathbb R^{784}$ and MNIST dataset \cite{lecun2010mnist}.
Our algorithm is compared with \textbf{Network Simplex} Algorithm \cite{orlin1997polynomial} and \textbf{Sinkhorn} Algorithm \cite{cuturi2013sinkhorn} by Python Optimal Transport Library \cite{flamary2021pot}.
For test \textit{Weight Modification}, we randomly select a pair of nodes in the system, and send flow of which amount is in the range of $\pm 1\%$ of their node weight difference, repeat for 100 times;
for test \textit{Point Deletion \& Insertion}, we first delete 100 nodes with 0 weight in the current data set to return it to the pool, and insert 100 new nodes in the system.
\begin{figure}[htbp]
      \centering
    \includegraphics[width = \linewidth]{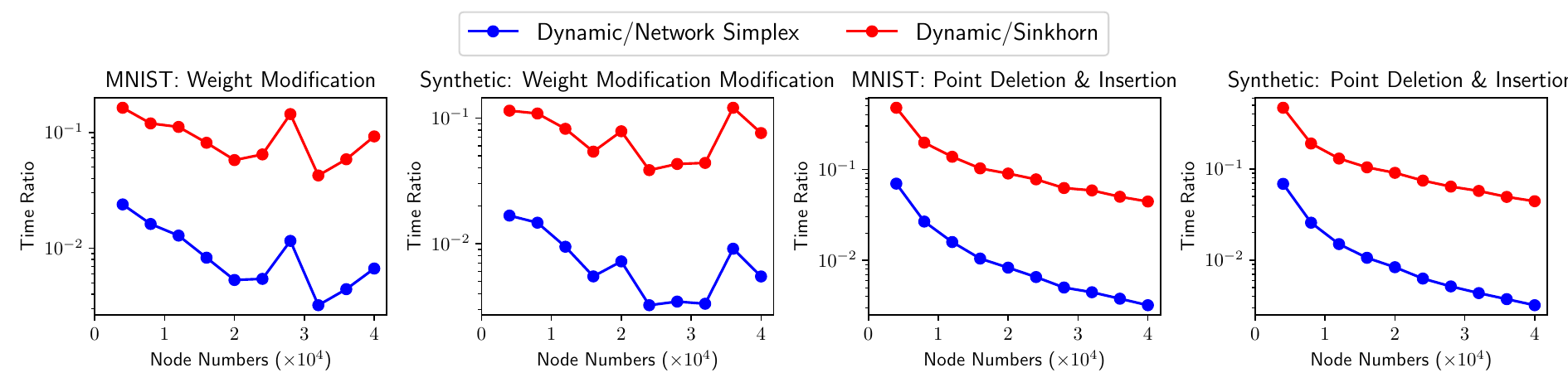}
    \caption{Supplementary Experiments}
    \label{fig:supplementresult}
\end{figure}

Figure \ref{fig:supplementresult} demonstrates our experiment result,
where the time ratio of our algorithm and their algorithm are plotted.
From the figure, the asymptotic advantage of our algorithm is obvious.
where $T_{\rm our}$ is the running time of each dynamic operation of our algorithm, while $T_{\rm net}$ and $T_{\rm sinkhorn}$ are the static running time of Network Simplex algorithm and Sinkhorn algorithm respectively. Our algorithm for handling point insertion only takes less than $1\%$ of time compared to static algorithms when $\abs{V}$ reaches $40000$.
Also the \textsc{Insert} and \textsc{Delete} procedures perform much more stable than space position modification and weight modification procedures,
regardless of which dataset we use.
This matches the theory that procedure \textsc{Insert} and \textsc{Delete} is strictly linear,
while the performance of normal space position modification and weight modification procedures depends on the number of simplex iterations it experiences.
Though procedure \textsc{Insert} and \textsc{Delete} must be combined with weight modification procedures to obtain the result,
Figure \ref{fig:supplementresult} shows that procedure \textsc{Insert} and \textsc{Delete} are not a great overhead.
\end{appendices}

\bibliographystyle{plainnat}
\bibliography{dynamic_ot_preprint}


\end{document}